\documentclass[a4paper]{article}
\usepackage[margin=1in]{geometry}
\usepackage[T1]{fontenc}
\usepackage{microtype}

\usepackage{hyperref}
\usepackage{amsmath,amsthm,amsfonts}
\usepackage{cleveref}
\usepackage{xspace}
\usepackage{url}
\usepackage{algorithm}
\usepackage{algpseudocode}
\usepackage{appendix}
\usepackage{enumitem}

\usepackage{multirow}
\usepackage{booktabs}
\usepackage{graphicx}
\usepackage{subcaption}
\usepackage{balance}

\hypersetup{colorlinks=true,allcolors=blue}

\theoremstyle{plain}
\newtheorem{theorem}{Theorem}[section]
\newtheorem{lemma}[theorem]{Lemma}
\newtheorem{claim}[theorem]{Claim}
\newtheorem*{claim*}{Claim}
\newtheorem{corollary}[theorem]{Corollary}

\newtheorem{fact}[theorem]{Fact}

\theoremstyle{definition}
\newtheorem{definition}[theorem]{Definition}

\theoremstyle{remark}
\newtheorem{remark}[theorem]{Remark}

\newcommand{\ProblemName}[1]{\textsc{#1}}
\newcommand{\MaxCut}{\ProblemName{Max-Cut}\xspace}

\DeclareMathOperator*{\E}{\mathbb{E}}

\DeclareMathOperator{\supp}{supp}

\def\RR{{\mathbb{R}}}

\DeclareMathOperator{\poly}{poly}

\DeclareMathOperator{\dist}{dist}

\DeclareMathOperator{\cut}{cut}

\DeclareMathOperator{\anc}{anc}

\DeclareMathOperator*{\nil}{\perp}
\newcommand{\Dall}{D_{\mathrm{all}}}
\newcommand{\Xheavy}{X_{\mathrm{heavy}}}
\newcommand{\Xlight}{X_{\mathrm{light}}}
\newcommand{\len}{\mathrm{len}}

\newcommand{\calG}{\mathcal{G}}

\newcommand{\Xext}{X^{\mathrm{ext}}}

\let\epsilon\varepsilon

\author{
  Xiaoyu Chen\thanks{Email: \texttt{yuchen21@stu.pku.edu.cn}
  }\\
  Peking University
  \and Shaofeng H.-C. Jiang\thanks{
    Research partially supported by a national key R\&D program of China No. 2021YFA1000900,
    a startup fund from Peking University, and the Advanced Institute of Information Technology, Peking University.
    Email: \texttt{shaofeng.jiang@pku.edu.cn}
  }\\
  Peking University
  \and Robert Krauthgamer\thanks{Work partially supported by ONR Award N00014-18-1-2364,
   by a Weizmann-UK Making Connections Grant,
   by a Minerva Foundation grant,
   and the Weizmann Data Science Research Center.
    Email: \texttt{robert.krauthgamer@weizmann.ac.il}
  }\\
  Weizmann Institute of Science
}

\begin{document}

\title{Streaming Euclidean \MaxCut: Dimension vs Data Reduction}
\date{}

\maketitle

\begin{abstract}
\MaxCut is a fundamental problem 
that has been studied extensively in various settings.
We design an algorithm for Euclidean \MaxCut,
where the input is a set of points in $\RR^d$,
in the model of dynamic geometric streams, 
where the input $X\subseteq [\Delta]^d$
is presented as a sequence of point insertions and deletions.
Previously, Frahling and Sohler [STOC 2005]
designed a $(1+\epsilon)$-approximation algorithm 
for the low-dimensional regime, i.e., it uses space $\exp(d)$. 

To tackle this problem in the high-dimensional regime,
which is of growing interest, 
one must improve the dependence on the dimension $d$,
ideally to space complexity $\poly(\epsilon^{-1} d \log\Delta)$.
Lammersen, Sidiropoulos, and Sohler [WADS 2009] proved that Euclidean \MaxCut
admits dimension reduction with target dimension $d' = \poly(\epsilon^{-1})$.
Combining this with the aforementioned algorithm that uses space $\exp(d')$,
they obtain an algorithm whose overall space complexity is
indeed polynomial in $d$, but unfortunately exponential in $\epsilon^{-1}$.

We devise an alternative approach of \emph{data reduction},
based on importance sampling,
and achieve space bound $\poly(\epsilon^{-1} d \log\Delta)$,
which is exponentially better (in $\epsilon$) than the dimension-reduction approach.
To implement this scheme in the streaming model,
we employ a randomly-shifted quadtree to construct a tree embedding.  
While this is a well-known method, 
a key feature of our algorithm is that the embedding's distortion $O(d\log\Delta)$
affects only the space complexity,
and the approximation ratio remains $1+\epsilon$.
\end{abstract}

    \section{Introduction}
\label{sec:intro}

\MaxCut is a fundamental problem in multiple domains,
from constraint satisfaction (CSP) and linear equations to clustering.
It was studied extensively in many computational models and for types of inputs,
and many (nearly) tight bounds were obtained,
oftentimes leading the way to even more general problems.
For instance, in the offline setting,
\MaxCut admits a polynomial-time $0.878$-approximation for general graphs~\cite{DBLP:journals/jacm/GoemansW95},
and this approximation factor is tight under the Unique Games Conjecture~\cite{DBLP:journals/siamcomp/KhotKMO07}.
In contrast, if the input is a dense unweighted graph,
or a metric space (viewed as a weighted graph),
then a PTAS exists~\cite{DBLP:journals/rsa/VegaK00,VK01}.
In the graph-streaming setting,
$(1 + \epsilon)$-approximation can be obtained using $\tilde{O}(n)$ space~\cite{DBLP:conf/icalp/AhnG09},
and this space bound is tight~\cite{DBLP:conf/stoc/KapralovK19}.

However, the streaming complexity of \MaxCut is only partially resolved
in the geometric setting, i.e., for Euclidean points.
A known algorithm, due to Frahling and Sohler~\cite{DBLP:conf/stoc/FrahlingS05},
achieves $(1 + \epsilon)$-approximation but uses space $\exp(d)$,
which is prohibitive when the dimension is high.
Combining this algorithm with a dimension reduction result,
based on the Johnson-Lindenstrauss Lemma but specialized to \MaxCut
and has target dimension $\poly(\epsilon^{-1})$~\cite{DBLP:conf/wads/LammersenSS09,lammersen2011approximation},
one can achieve polynomial dependence on $d$,
but at the expense of introducing to the space complexity
an undesirable $\exp(\poly(\epsilon^{-1}))$-factor.
It was left open to obtain in the high-dimension regime
space complexity that is truly efficient, i.e., $\poly(\epsilon^{-1} d)$.

We answer this question by providing the first streaming algorithms
that achieve $(1+\epsilon)$-approximation for \MaxCut using space $\poly(d \epsilon^{-1})$.
We consider the setting of \emph{dynamic geometric streams},
introduced by Indyk~\cite{Indyk04},
where the input is a dataset $X \subseteq [\Delta]^d$
that is presented as a stream of point insertions and deletions.
The goal of the algorithm is to approximate (multiplicatively)
the so-called \MaxCut value, defined as
\[
  \MaxCut(X) := \max_{S\subseteq X} \sum_{x \in S, y \in X \setminus S}\|x -y\|_2
\]
(see \Cref{sec:prelim} for general metric spaces).
We say that $\MaxCut(X)$ is \emph{$\alpha$-approximated},
for $\alpha \geq 1$, by a value $\eta\ge 0$
if $\MaxCut(X)/\alpha \leq \eta \leq \MaxCut(X)$.\footnote{We will actually aim at $\eta \in (1\pm\epsilon) \cdot \MaxCut(X)$,
  for $0<\epsilon<1/2$,
  which can be scaled to achieve $(1+O(\epsilon))$-approximation.
}
We assume throughout that $X$ contains distinct points (and is not a multiset),
hence $n:=|X|\leq |\Delta|^d$.
In the \emph{high-dimension regime},
algorithms can use at most $\poly(d \log \Delta)$ bits of space,
which is polynomial in the number of bits required to represent a point in $[\Delta]^d$,
and also allows counting to $n \leq |\Delta|^d$. 
In the \emph{low-dimension regime},
algorithms may have bigger space complexity, e.g., exponential in $d$.

A central challenge in the area of geometric streaming algorithms
is to achieve good accuracy (approximation) in the high-dimension regime.
Indeed, algorithms for many basic problems
(like diameter, minimum spanning tree, facility location, and \MaxCut),
achieve good approximation, say for concreteness $O(1)$ or even $1+\epsilon$,
using space that is exponential in $d$
\cite{AHV04, DBLP:journals/algorithmica/Zarrabi-Zadeh11, DBLP:conf/stoc/FrahlingS05, DBLP:journals/ijcga/FrahlingIS08,DBLP:conf/esa/LammersenS08, DBLP:conf/soda/CzumajLMS13, DBLP:conf/icalp/CzumajJKV22}.
In contrast, algorithms that use space polynomial in $d$ are fewer
and they typically achieve far worse approximation ratio
\cite{Indyk04,DBLP:conf/stoc/ChenJLW22,CJKVW22,WY22},
and obtaining $O(1)$-approximation remains open.
In particular, Indyk~\cite{Indyk04} tackled the high-dimension regime
using a technique of randomized tree embedding, which is rather general and economical in space,
but unfortunately distorts distances by a factor of $O(d\log\Delta)$
that goes directly into the approximation ratio.
Attempts to improve the approximation ratio had only limited success so far;
for example, the algorithm of~\cite{AS15} (for diameter)
works only in insertion-only streams,
and the algorithms of~\cite{DBLP:conf/stoc/ChenJLW22,CJKVW22}
(for MST and for facility location)
fall short of the desired $O(1)$-approximation in one pass.

\subsection{Our Results}
We bypass the limitation of dimension reduction via a data reduction approach,
and design a streaming algorithm that $(1+\epsilon)$-approximates \MaxCut
using $\poly(\epsilon^{-1} d \log\Delta)$ space,
thereby closing the gap of high dimension (for \MaxCut).
Our approach works not only under Euclidean norm,
but also when distances are calculated using $\ell_p$ norm, for $p\ge1$. 

\paragraph{Data Reduction via Importance Sampling}
We present an algorithm that is based on the data-reduction approach,
namely, it uses the dataset $X$ to construct a small instance $X'$
that has a similar \MaxCut value,
then solve \MaxCut on it optimally and report this value $\MaxCut(X')$.
Following a common paradigm, $X'$ is actually a re-weighted subset of $X$,
that is picked by non-uniform sampling from $X$, known as importance sampling.

\begin{theorem}[Streaming $\MaxCut$ in $\ell_p$ Norm]
    \label{thm:streaming_intro}
    There is a randomized streaming algorithm that,
    given $0 < \epsilon < 1/2$, $p \geq 1$,
    integers $\Delta, d \geq 1$,
    and an input dataset $X \subseteq [\Delta]^d $ presented as a dynamic stream,
    uses space $\poly(\epsilon^{-1} d \log \Delta)$
    and reports an estimate $\eta>0$ that with probability at least $2/3$
    is a $(1 + \epsilon)$-approximation to $\MaxCut(X)$ in $\ell_p$.
\end{theorem}

This data-reduction approach was previously used for several clustering problems.
For $k$-median and related problems,
such an instance $X'$ is often called a coreset,
and there are many constructions,
see e.g.~\cite{DBLP:conf/stoc/Har-PeledM04,DBLP:conf/stoc/FrahlingS05,DBLP:conf/stoc/FeldmanL11,BFLSY17,DBLP:journals/corr/abs-1802-00459,DBLP:journals/siamcomp/FeldmanSS20}.
Earlier work~\cite{Schulman00} has designed importance-sampling based 
algorithms for a problem closely related to \MaxCut,
but did not provide a guarantee that $X$ and $X'$ have a similar \MaxCut value
(see \Cref{sec:related} for a more detailed discussion).
Recently, importance sampling was used to design streaming algorithms
for facility location in high dimension~\cite{CJKVW22},
although their sample $X'$ is not an instance of facility location.
Overall, this prior work is not useful for us, 
and we have to devise our own sampling distribution,
prove that it preserves the \MaxCut value,
and design a streaming algorithm that samples from this distribution.

The approximation ratio $1 + \epsilon$ in \Cref{thm:streaming_intro} is essentially the best one can hope for using small space,
because finding the \MaxCut value exactly, even in one dimension,
requires $\Omega(\Delta)$ space, as shown in \Cref{claim:lb_exact}.
Compared with the dimension-reduction approach (based on~\cite{DBLP:conf/stoc/FrahlingS05}), 
our \Cref{thm:streaming_intro} has the advantage that it works
for all $\ell_p$ norms ($p \geq 1$) and not only $\ell_2$.
The result of~\cite{DBLP:conf/stoc/FrahlingS05} is stronger in another aspect,
of providing a ``cut oracle'',
i.e., an implicit representation of the approximately optimal cut
that can answer the side of the cut that each data point $x\in X$
(given as a query) belongs to.
This feature extends also to high dimension,
as it is easy to combine with the dimension reduction. 
For completeness,
we give a (somewhat simplified) proof of the dimension reduction in \Cref{sec:jl},
followed by a formal statement of this ``cut oracle'' in \Cref{cor:cut_oracle}.
It remains open to design a streaming algorithm that computes
such an implicit representation using space $\poly(\epsilon^{-1} d \log \Delta)$.

\subsection{Technical Overview}

In order to estimate \MaxCut using importance sampling,
we must first identify a sampling distribution for which
$\MaxCut(X')$ indeed approximates $\MaxCut(X)$,
and then we have to design a streaming algorithm
that samples from this distribution.

\paragraph{Sampling Probability}
One indication that geometric \MaxCut admits data reduction by sampling 
comes from the setting of dense unweighted graphs,
where it is known that \MaxCut can be $(1+\epsilon)$-approximated
using a \emph{uniform sample} of $O(\epsilon^{-4})$ vertices,
namely, by taking the \MaxCut value in the induced subgraph and scaling
appropriately~\cite{DBLP:journals/jcss/AlonVKK03,DBLP:journals/jacm/RudelsonV07}
(improving over \cite{GGR96}). 
However, sampling points uniformly clearly cannot work in the metric case --
if a point set $X$ has many co-located points
and a few distant points that are also far from each other,
then uniform sampling from $X$ is unlikely to pick the distant points,
which have a crucial contribution to the \MaxCut value. 
It is therefore natural to employ importance sampling,
i.e., sample each point with probability proportional to its contribution.
The contribution of a point $x$ to the \MaxCut value is difficult to gauge,
but we can use instead a simple proxy --
its total distance to all other points $q(x) := \sum_{y \in X} \dist(x, y)$,
which is just its contribution to twice the total edge weight
$\sum_{x,y\in X} \dist(x, y)$.
For any fixed cut in $X$,
this sampling works well and the estimate will (likely) have additive error
$\epsilon \sum_{x,y\in X} \dist(x, y) = \Theta(\epsilon) \cdot\MaxCut(X)$.
While the analysis is straightforward for a fixed cut, say a maximum one,
proving that the sampling preserves the \MaxCut value
is much more challenging, as one has to argue about all possible cuts. 

We show in \Cref{thm:sampling} that $O(\epsilon^{-4})$ independent samples
generated with probability proportional to $q(x)$ preserve the \MaxCut value.
This holds even if the sampling probabilities are dampened by a factor $\lambda \geq 1$,
at the cost of increasing the number of samples by a $\poly(\lambda)$ factor.
To be more precise,
it suffices to sample from any probability distribution $\{p_x:\ x\in X\}$,
where $p(x) \ge \frac{1}{\lambda} \frac{q(x)}{\sum_{y\in X} q(y)}$ for all $x\in X$.
A small technicality is that we require the sampling procedure
to report a random sample $x^*$ together with its corresponding $p(x^*)$,
in order to re-weight the sample $x^*$ by factor $1/p(x^*)$,
but this extra information can often be obtained using the same methods.

This sampling distribution,
i.e., probabilities proportional to $\{q(x):\ x\in X\}$,
can be traced to two prior works.
Schulman~\cite{Schulman00} used essentially the same probabilities,
but his analysis works only for one fixed cut,
and extends to a multiple cuts by a union bound.\footnote{We gloss over slight technical differences,
  e.g., he deals with squared Euclidean distances,
and his sampling and re-weighting processes are slightly different. 
}
The exact same probabilities were also used by~\cite{VK01}
as weights to convert a metric instance to a dense unweighted graph, 
and thereby obtain a PTAS (without sampling or any data reduction).

In fact, our proof combines several observations from~\cite{VK01} 
about a uniform sample of $O(\epsilon^{-4})$ vertices in unweighted graphs~\cite{DBLP:journals/jcss/AlonVKK03,DBLP:journals/jacm/RudelsonV07}.
In a nutshell, we relate
sampling from $X$ proportionally to $\{q(x):\ x\in X\}$ 
with sampling uniformly from a certain dense unweighted graph,
whose cut values corresponds to those in $X$,
and we thereby derive a bound on the \MaxCut value of the sample $X'$.

\paragraph{Streaming Implementation}
Designing a procedure to sample proportionally to $\{q(x):\ x\in X\}$, 
when $X$ is presented as a stream,
is much more challenging and is our main technical contribution
(\Cref{lem:importance_sampling_algorithm}).
The main difficulty is that standard tools for sampling from a stream,
such as $\ell_p$-samplers \cite{MW10,JST11,JW21},
are based on the frequency vector 
and oblivious to the geometric structure of $X$.
Indeed, the literature lacks geometric samplers,
which can be very useful when designing geometric streaming algorithms.
Our goal of sampling proportionally to $q(x)$,
which is the total distance to all other points in $X$,
seems like a fundamental geometric primitive, 
and therefore our sampler (\Cref{lem:importance_sampling_algorithm})
is a significant addition to the geometric-streaming toolbox,
and may be of independent interest.

\paragraph{High-Level Picture}
At a high level,
our sampling procedure is based on a randomly-shifted quadtree $T$
that is defined on the entire input domain $[\Delta]^d$ and captures its geometry.
This randomly-shifted quadtree $T$ is data oblivious,
and thus can be picked (constructed implicitly)
even before the stream starts (as an initialization step), 
using space $O(\poly(d \log \Delta))$.
This technique was introduced by Indyk~\cite{Indyk04},
who noted that the quadtree essentially defines a tree embedding 
with expected distortion of distances $O(d\log\Delta)$. 
Unfortunately, the distortion is fundamental to this approach,
and directly affects the accuracy (namely, goes into the approximation ratio)
of streaming algorithms that use this technique~\cite{Indyk04,DBLP:conf/stoc/ChenJLW22}.\footnote{A randomly-shifted quadtree was used also
  in Arora's approximation algorithm for TSP~\cite{Arora98},
  but as the basis for dynamic programming rather than as a tree embedding,
  and similarly in streaming applications of this technique~\cite{DBLP:conf/icalp/CzumajJKV22}.
}
While our approach still suffers this distortion,
we can avoid its effect on the accuracy;
instead, it affects the importance-sampling distribution,
namely, the dampening factor $\lambda$ grows by factor $O(d\log\Delta)$,
which can be compensated by drawing more samples.
This increases the space complexity moderately, which we can afford,
and overall leads to $(1 + \epsilon)$-approximation
using space $\poly(\epsilon^{-1} d \log \Delta)$.

\paragraph{Comparison to Other Sampling Approaches}
A different importance sampling method 
was recently designed in~\cite{CJKVW22} for facility location in high dimension.
Their sampling procedure relies on a geometric hashing (space partitioning),
and completely avoids a quadtree.\footnote{A key parameter in that hashing, called consistency, turns out to be $\poly(d)$,
and unfortunately affects also the final approximation ratio for facility location,
and not only the importance-sampling parameter $\lambda$ (and thus the space).
}
Our sampling technique may be more easily applicable to other problems,
as it uses a more standard and general tool of a randomly-shifted quadtree.
Another importance-sampling method
was recently designed in~\cite{MahabadiRWZ20} in the context of matrix streams.
Geometrically, their importance-sampling can be viewed as
picking a point (row in the matrix) proportionally to its length, 
but after the points are projected, at query time, onto a subspace. This sampling procedure is very effective for linear-algebraic problems,
like column-subset selection and subspace approximation,
which can be viewed also as geometric problems.

Previously, quadtrees were employed in geometric sampling results,
but mostly a fixed quadtree and not a randomly-shifted one, 
and to perform \emph{uniform sampling} over its non-empty squares at each level,
as opposed to our importance sampling~\cite{DBLP:conf/stoc/FrahlingS05,DBLP:journals/ijcga/FrahlingIS08,DBLP:conf/icalp/CzumajJKV22}.
Furthermore, in~\cite{DBLP:journals/ijcga/FrahlingIS08}, 
uniform sampling was augmented to report also all the points nearby the sampled points, and this was used to estimate the number of connected components in a metric threshold graph.

\paragraph{Sampling on Quadtree: Bypassing the Distortion}
Finally, we discuss how to perform the importance sampling, 
based on a randomly-shifted quadtree $T$,
but without the $O(d\log\Delta)$ distortion going into the approximation ratio. 
As explained above, our importance sampling (\Cref{thm:sampling}) requires that 
every point $x\in X$ is sampled with some probability 
$p(x) \geq \frac{1}{\lambda} \frac{q(x)}{Q}$ for $Q := \sum_{y\in X} q(y)$,
and the dampening factor $\lambda$ can be up to $O(\poly(d \log \Delta))$. 
We will exploit the fact that there is no upper bound on the probability $p(x)$.
As usual, a randomly-shifted quadtree $T$ is obtained by
recursively partitioning of the grid $[2\Delta]^d$,
each time into $2^d$ equal-size grids, which we call squares,
and building a tree whose nodes are all these squares,
and every square is connected to its parent by an edge whose weight
is equal to that square's Euclidean diameter.
Furthermore, this entire partitioning is shifted
by a uniformly random $v_{\textrm{shift}}\in [\Delta]^d$.
(See \Cref{sec:proof_importance_sampling} for details.) 
The key is that every grid point $x\in [\Delta]^d$ is represented by a tree leaf,
and thus $T$ defines distances on $[\Delta]^d$, also called a tree embedding. 
Define $q_T$ and $Q_T$ analogously to $q$ and $Q$,
but using the \emph{tree distance},
that is, $q_T(x)$ is the total tree distance from $x$ to all other points in $X$,
and $Q_T := \sum_{y\in X} q_T(y)$. 

The tree-embedding guarantees, for all $x\in[\Delta]^d$, 
that $q_T(x)$ overestimates $q(x)$ (with probability $1$), 
and that $\E_T[q_T(x)] \leq O(d \log \Delta) \cdot q(x)$.
It thus suffices to have one event, $Q_T \leq O(d \log \Delta) \cdot Q$,
which happens with constant probability by Markov's inequality,
and then we would have
$\frac{q_T(x)}{Q_T} \geq \frac{1}{O(d \log \Delta)} \cdot \frac{q(x)}{Q}$
simultaneously for all $x\in X$.
As mentioned earlier, the algorithm picks (constructs implicitly) $T$
even before the stream starts,
and the analysis assumes the event mentioned above happens.
In fact, the rest of the sampling procedure works correctly for arbitrary $T$,
i.e., regardless of that event.
We stress, however, that our final algorithm actually needs to generate
multiple samples that are picked independently from the same distribution,
and this is easily achieved by parallel executions 
that share the same random tree $T$ but use independent coins in all later steps. 
We thus view this $T$ as a preprocessing step that is run only once,
particularly when we refer to samples as independent.

\paragraph{A Two-level Sampling by Using Tree Structure}
The remaining challenge is to sample (in a streaming fashion)
with probability proportional to $q_T$ for a fixed quadtree $T$.
To this end, we make heavy use of the structure of the quadtree.
In particular, the quadtree $T$ has $O(d \log \Delta)$ levels,
and every data point $x\in X$ forms a distinct leaf in $T$.
Clearly, each internal node in $T$ represents a subset of $X$
(all its descendants in $T$, that is, the points of $X$ inside its square).
At each level $i$ (the root node has level $1$),
we identify a level-$i$ \emph{heavy node} $h_i$ that contains
the maximum number of points from $X$ (breaking ties arbitrarily).
We further identify a \emph{critical level} $k$,
such that $h_1, \ldots, h_k$ (viewed as subsets of $X$),
all contain more than $0.5 |X|$ points, but $h_{k+1}$  does not.
This clearly ensures that $h_1, \ldots, h_k$ forms a path.
Let $X(h_i) \subseteq X$ denote the subset of $X$ represented by node $h_i$.
The tree structure and the path property of $(h_1, \ldots, h_k)$ guarantee that,
for each $i < k$,
all points $x\in X(h_{i}) \setminus X(h_{i + 1})$ have roughly the same $q_T(x)$ values.
For the boundary case $i = k$, a similar claim holds for $x \in X(h_{k})$.
Hence, a natural algorithm is to employ two-level sampling:
first draw a level $i^* \leq k$,
then draw a uniform sample from $X(h_{i^*}) \setminus X(h_{i^* + 1})$.
The distribution of sampling $i^*$ also requires a careful design,
but we omit this from the current overview,
and focus instead on how to draw a uniform sample from $X(h_i) \setminus X(h_{i + 1})$.

Unfortunately, sampling from $X(h_i) \setminus X(h_{i+1})$ still requires $\Omega(\Delta)$ space in the streaming model, even for $d = 1$.
(We prove this reduction from INDEX in \Cref{claim:lb_sampling_light}.)
In fact, it is even hard to determine whether $X(h_i) \setminus X(h_{i + 1}) = \emptyset$;
the difficulty is that $h_i$ is not known in advance, otherwise it would be easy. 
We therefore relax the two-level sampling,
and replace sampling from $X(h_i) \setminus X(h_{i + 1})$,
with its superset $X \setminus X(h_{i + 1})$.
This certainly biases the sampling probability,
in fact some probabilities might change by an unbounded factor (\Cref{rem:ChangeOrderOfSum}),
nevertheless we prove that the increase in the dampening parameter $\lambda$
is bounded by an $O(\log\Delta)$ factor (\Cref{lem:prob}),
which we can still afford.
This part crucially uses the tree and the path property of $(h_1, \ldots, h_k)$.

\paragraph{Sampling from Light Parts}
The final remaining step is to sample from $X \setminus X(h_{i})$
for a given $i\leq k$ (\Cref{lemma:streaming_light}).
We can assume here $X(h_i)$ contains more than $0.5 |X|$ points,
and thus $X \setminus X(h_i)$ is indeed the ``light'' part,
containing few (and possibly no) points.
To let the light points ``stand out'' in a sampling,
we hash the \emph{tree nodes} (instead of the points) randomly into two buckets.
Since $|X(h(i))| > 0.5 |X|$,
the heavy node $h_i$ always lies in the bucket that contains more points,
and we therefore sample only from the light bucket.
(To implement this in streaming,
we actually generate two samples, one from each bucket,
and in parallel estimate the two buckets' sizes
to know which of the two samples to use.)
Typically, the light bucket contains at least half of the points
from $X \setminus X(h_i)$, which is enough. 
Overall, this yields a sampling procedure that uses space $\poly(\epsilon^{-1}d \log \Delta)$.

\subsection{Related Work}
\label{sec:related}

Geometric streaming in the low-dimension regime was studied much more extensively
than in the high-dimensional case that we investigate here.
In~\cite{DBLP:conf/stoc/FrahlingS05}, apart from the $O(\epsilon^{-d}\poly\log \Delta)$-space $(1 + \epsilon)$-approximation for \MaxCut that we already mentioned,
similar results were obtained also for $k$-median, maximum spanning tree, maximum matching and similar maximization problems.
For minimum spanning tree,
a $(1 + \epsilon)$-approximation $O((\epsilon^{-1}\log\Delta)^d)$-space 
algorithm was devised in \cite{DBLP:journals/ijcga/FrahlingIS08}, 
alongside several useful techniques for geometric sampling in low dimension.
In~\cite{DBLP:conf/focs/AndoniBIW09}, an $O(\epsilon^{-1})$-approximation
$\tilde{O}(\Delta^{\epsilon})$-space streaming algorithm was obtained
for computing earth-mover distance in dimension $d = 2$.
For facility location in dimension $d=2$,
a $(1 + \epsilon)$-approximation $\poly(\epsilon^{-1}\log \Delta)$-space 
algorithm was designed in~\cite{DBLP:conf/soda/CzumajLMS13}.
Recently, Steiner forest (a generalization of Steiner tree,
which asks to find a minimum-weight graph that connects $k$ groups of points),
was studied in~\cite{CJKVW22} for dimension $d=2$,
and they obtain $O(1)$-approximation using space $(\poly(k \log \Delta))$.
 
Our sampling distribution may seem reminiscent of
an importance sampling procedure devised by Schulman~\cite{Schulman00},
however his results and techniques are not useful in our context.
First, the problem formulation differs,
as it asks to approximate the complement objective of sum of distances inside each cluster (which is only stronger than our MAX-CUT objective),
but the objective function sums squared Euclidean (rather than Euclidean) distances. 
Second, his algorithm is for the offline setting
and is not directly applicable in streaming.
Third, his analysis does not provide a guarantee on $\MaxCut(X')$,
but rather only on certain cuts of $X'$ (not all of them),
and his approximation guarantee includes a non-standard twist.

     \section{Preliminaries}
\label{sec:prelim}
Consider a metric space $(V, \dist)$.
The Euclidean case is $V = \mathbb{R}^d$ and $\dist = \ell_2$,
and the $\ell_p$ case is $V = \mathbb{R}^d$ and $\dist = \ell_p$.
For $X \subseteq V$, the cut function $\cut_X : 2^X \to \mathbb{R}$
is defined as
$\cut_X(S) := \sum_{x\in S, y \in X \setminus S} \dist(x, y)$.
The \MaxCut value of a dataset $X \subseteq V$ is defined as
\[
    \MaxCut(X) := \max_{S \subseteq X}{\cut_X(S)}.
\]
We shall use the following standard tools as building blocks in our algorithms.
Recall that a turnstile (or dynamic) stream can contain insertions and deletions of items from some domain $[N]$,
and it naturally defines a frequency vector $x\in\mathbb{R}^N$,
where every possible item has a coordinate that counts its net frequency, i.e., insertions minus deletions.
In general, this model allows frequencies to be negative.
However, in our setting where $X\subset [\Delta]^d$ is presented as a dynamic geometric stream,
then frequency vector has $\Delta^d$ coordinates all in the range $\{0,1\}$,
and is just the incidence vector of $X$.

\begin{lemma}[$\ell_0$-Norm Estimator~\cite{DBLP:conf/pods/KaneNW10}]
    \label{lem:l0estimator}
    There exists a streaming algorithm that,
    given $0 < \epsilon, \delta< 1$,
    integers $N, M \geq 1$,
    and a frequency vector $x \in [-M, M]^N$ presented as a turnstile stream,
    where we denote its support by $X := \{ i \in [N] :\ x_i \neq 0 \}$,
    uses space $\poly(\epsilon^{-1}\log(\delta^{-1} MN))$
    to return $r^* \geq 0$,
    such that $\Pr[ r^* \in (1 \pm \epsilon) |X| ] \geq 1-\delta$.
\end{lemma}

\begin{lemma}[$\ell_0$-Sampler~\cite{JST11}]
    \label{lem:l0sampler}
    There exists a streaming algorithm that,
    given $0 < \delta< 1$,
    integers $N, M \geq 1$,
    and a frequency vector $x \in [-M, M]^N$ presented as a turnstile stream,
    where we assume that its support $X := \{ i \in [N] : x_i \neq 0 \}$ is non-empty,
    uses space $\poly\log(\delta^{-1} MN)$,
    to return a sample $i^* \in X \cup \{\perp\}$,
    such that with probability at least $ 1 - \delta$,
    \[
        \forall i \in X, \qquad \Pr[i^* = i] = \frac{1}{|X|}.
    \]
\end{lemma}

     \section{Approximating \MaxCut by Importance Sampling}
\label{sec:sampling}
In this section, we consider a general metric space $(V, \dist)$ (which includes Euclidean spaces by setting $V = \mathbb{R}^d$ and $\dist = \ell_2$),
and show that a small importance-sample $S$ on a dataset $X \subseteq V$
may be used to estimate $\MaxCut(X)$ by simply computing $\MaxCut(S)$.

\paragraph{Point-weighted Set~\cite{VK01}}
Since we apply importance sampling (as opposed to using a uniform probability for every point),
the sampled points need to be re-weighted so that the estimation is unbiased.
Hence, we consider the notion of \emph{point-weighted} sets.
This notion was first considered in~\cite{VK01} to reduce the metric \MaxCut
to \MaxCut in (dense) weighted graphs.
Specifically, a point-weighted set $S \subseteq V$ is a subset of $V$ that is associated with a point-weight function $w_S : S \to \mathbb{R}_+$.
For a point-weighted set $S$, the distance $\dist_S(x, y)$
between $x, y \in S$ is also re-weighted such that $\dist_S(x, y) := \frac{\dist(x, y)}{w_S(x) w_S(y)}$.
Under this weighting, when an edge $\{x, y\}$ appears in a cut,
its contribution is still accounted as $w_S(x) \cdot w_S(y) \cdot \dist_S(x, y) = \dist(x, y)$.

We prove (in \Cref{thm:sampling}) that for every dataset $X \subseteq V$,
if one construct a point-weighted subset $S \subseteq X$
by drawing i.i.d. samples from a distribution on $X$,
where each $x \in X$ is sampled proportional to
    $q(x) = \sum_{y \in X}{\dist(x, y)}$
which is the sum of distances to other points in $X$,
up to an error factor of $\lambda$,
then $\MaxCut(S)$ is $(1 + \epsilon)$-approximation to $\MaxCut(X)$
with high probability.

\begin{theorem}\label{thm:sampling}
      Given $\varepsilon,\delta > 0,\lambda \ge 1$, metric space $(V, \dist)$
      and dataset $X \subseteq V$,
      let $\mathcal{D}$ be a distribution $(p_x : x \in X)$ on $X$ such that
      $\forall x \in X, p_x \geq \frac{1}{\lambda} \cdot \frac{q(x)}{Q}$,
      where $q(x) = \sum_{y \in X} \dist(x, y)$ and $Q = \sum_{x \in X}q(x)$.
      Let $S$ be a point-weighted set that is obtained by an i.i.d. sample of
      $m \geq 2$  points from $\mathcal{D}$,
      weighted by $w_S(x) := \hat{p}_x$ such that
      $p_x \leq \hat{p}_x \leq (1 + \epsilon) \cdot p_x$.
      If $m \geq O(\epsilon^{-4}\lambda^{-8})$,
      then with probability at least $0.9$,
      the value $\frac{\MaxCut(S)}{m^2} $ is a $(1 + \epsilon)$-approximation to $\MaxCut(X)$.
\end{theorem}

The $O(\epsilon^{-4})$ dependence on $\epsilon$ of \Cref{thm:sampling}
matches a similar $O(\epsilon^{-4})$ sampling complexity bound for unweighted graphs in~\cite{DBLP:journals/jcss/AlonVKK03,DBLP:journals/jacm/RudelsonV07}\footnote{A slightly weaker bound of $O(\epsilon^{-4}\poly\log(\epsilon^{-1}))$ was obtained in~\cite{DBLP:journals/jcss/AlonVKK03},
and \cite{DBLP:journals/jacm/RudelsonV07} gave an improved technical lemma
which can be directly plugged into~\cite{DBLP:journals/jcss/AlonVKK03} to obtain the $O(\epsilon^{-4})$ bound.}.
To the best of our knowledge, this $O(\epsilon^{-4})$ is the state-of-the-art even for the case of unweighted graphs.
Although our proof of \Cref{thm:sampling} is obtained mostly by using the bound in~\cite{DBLP:journals/jcss/AlonVKK03,DBLP:journals/jacm/RudelsonV07} as a black box,
the generalization to metric spaces, as well as the allowance of $\lambda$
which is the error in sampling probability, is new.

\subsection{Proof of \Cref{thm:sampling}}
As mentioned, the plan is to apply the sampling bound proved in~\cite{DBLP:journals/jcss/AlonVKK03,DBLP:journals/jacm/RudelsonV07},
which we restate in \Cref{lemma:eps4}.
In fact, the original statement in~\cite{DBLP:journals/jcss/AlonVKK03,DBLP:journals/jacm/RudelsonV07}
was only made for unweighted graphs, i.e., edge weights in $\{0, 1\}$.
However, we observe that only the fact that the edges weights are between $[0, 1]$ was used in their proof.
Hence, in our statement of \Cref{lemma:eps4} we make this stronger claim of $[0, 1]$ edge weights.

Here, for a graph $G(V, E)$ with weight function $\len_G : E \to \mathbb{R}$ ($\len_G(\cdot) = 1$ for unweighted graphs),
we define the cut function for $S \subseteq X \subseteq V$ (as well as \MaxCut) similarly,
as $\cut_X(S) = \sum_{x \in S, y \in X \setminus S : \{x, y\} \in E} \len_G(x, y)$.

\begin{lemma}[\cite{DBLP:journals/jcss/AlonVKK03,DBLP:journals/jacm/RudelsonV07}]
    \label{lemma:eps4}
    Consider a weighted graph $G(V, E)$ with weights in $[0, 1]$.
    Let $D \subseteq V$ be a uniformly independent sample from $V$ (possibly with repetitions) of $O(\epsilon^{-4})$ points.
    Then with probability at least $0.9$,
    \[
        \left|\frac{1}{|D|^2}\MaxCut(D) - \frac{1}{|V|^2}\MaxCut(V)\right| \leq \epsilon.
    \]
\end{lemma}

Hence, our plan is to define an auxiliary graph $G'$ whose edge weights are in $[0, 1]$,
such that our importance sampling may be interpreted as a uniform sampling from vertices in $G'$.
Eventually, our sampling bound would follow from \Cref{lemma:eps4}.

\paragraph{Defining Auxiliary Graph}
Since we focus on approximate solutions, we can assume that $p_x$'s ($x \in X$)
are of finite precision.
Then, let $N$ be a sufficiently large number such that for all $x \in X$,
$Np_x$ is an integer.
We define an auxiliary graph $G'(X', E' := X' \times X')$,
such that $X'$ is formed by copying each point $x \in X$ for $N p_x$ times, and edge weight
$\len_{G'} (x, y) := \frac{1}{4\lambda^2 Q} \cdot \frac{\dist(x, y)}{\hat{p}_x \hat{p}_y}$.
Clearly, if we let $x^*$ be a uniform sample from $X'$, then for every $x \in X$,
$\Pr[x^* = x] = p_x$.
Hence, this uniform sample $x^*$ is identically distributed as an importance-sample from distribution $\mathcal{D}$ on $X$.
Furthermore, for $x, y \in S$, it holds that
\[
    \dist_S(x, y) = \frac{\dist(x, y)}{\hat{p}_x \hat{p}_y}
    = 4\lambda^2 Q \cdot \len_{G'}(x, y).
\]
Hence, we conclude the following fact.
\begin{fact}
    \label{fact:coupling}
    Let $S' \subseteq X'$ be $m$ uniform samples from $X'$.
    Then, the value $4\lambda^2 Q \cdot \MaxCut(S')$ and $\MaxCut(S)$ are identically distributed.
\end{fact}
Therefore, it suffices to show the $S'$ from \Cref{fact:coupling}
satisfies that $4\lambda^2 Q \cdot \frac{\MaxCut(S')}{|S'|^2}$
is a $(1 + \epsilon)$-approximation to $\MaxCut(X)$, with constant probability.
Our plan is to apply \Cref{lemma:eps4}, but we first need to show, in \Cref{lemma:edge_weight_1}, that the edge weights of $G'$ are in $[0, 1]$,
and in \Cref{lemma:relate_value} that $\MaxCut(X')$
is a $(1 + \epsilon)$-approximation to $\MaxCut(X)$ up to a scaling.
\begin{lemma}
    \label{lemma:edge_weight_1}
    For all $x, y \in X'$, $\len_{G'}(x, y) \leq 1$.
\end{lemma}
\begin{proof}
    We need the following fact from~\cite{VK01} (which was proved in~\cite[Lemma 7]{VK01}).
    \begin{lemma}[\cite{VK01}]
        \label{lemma:dist_qx_qy}
        For all $x \in X$, it holds that
        \[
            \frac{\dist(x, y)}{q(x) q(x)} \leq \frac{4}{Q}.
        \]
    \end{lemma}
    Applying \Cref{lemma:dist_qx_qy},
    \[
        \len_{G'}(x, y)
        = \frac{1}{4\lambda^2 Q} \cdot \frac{\dist(x, y)}{\hat{p}_x \hat{p}_y}
        \leq \frac{1}{4\lambda^2 Q} \cdot \frac{ \dist(x, y) }{p_x p_y} \leq 1.
    \]
\end{proof}
\begin{lemma}
    \label{lemma:relate_value}
    $\frac{4\lambda^2 Q}{N^2}\MaxCut(X') \in (1 \pm \epsilon) \cdot  \MaxCut(X)$.
\end{lemma}
\begin{proof}
    Let $\widetilde{X}$ be a point-weighted set formed by
    re-weight points in $X$ with $w_{\widetilde{X}}(x) = N p_x$.
    The following lemma from~\cite{VK01} shows that $\MaxCut(X) = \MaxCut(\widetilde{X})$.
    \begin{lemma}[{\cite[Lemma 5 and Lemma 6]{VK01}}]
        Let $(U, \dist_U)$ be a metric space, and $W \subseteq U$ be a dataset.
        Suppose for every $x \in W$, $\mu_x > 0$ is an integer weight.
        Then the point-weighted set $W'$ obtained from re-weighting each point $x \in W$ by $\mu_x$, satisfies that
        \[
            \MaxCut(W') = \MaxCut(W).
        \]
    \end{lemma}
    Now, we observe that $\widetilde{X}$ can be naturally interpreted a weighted complete graph $\widetilde{G}$,
    where we copy $x \in X$ for $w_{\widetilde{X}}$ times to form the vertex set,
    and the edge length is defined as $\dist_{\widetilde{X}}(x, y)$.
    Notice that the vertex set of $\widetilde{G}$ is exactly $X'$,
    and that the edge length
    \[
        \dist_{\widetilde{X}}(x, y)
        = \frac{\dist(x, y)}{N^2 p_x p_y}
        \in (1\pm \epsilon) \cdot \frac{4\lambda^2 Q}{N^2} \cdot \len_{G'}(x, y).
    \]
    Therefore, we conclude that
    $\MaxCut(X) = \MaxCut(\widetilde{{X}}) \in (1 \pm \epsilon) \cdot \frac{4\lambda^2 Q}{N^2} \cdot \MaxCut(X')$.
    This finishes the proof.
\end{proof}
Now, we are ready to apply \Cref{lemma:eps4}.
Let $S' \subseteq X'$ such that $|S'| = O(\epsilon^{-4})$
be the resultant set by applying \Cref{lemma:eps4} with $G = G'$ (recalling that the promise of $[0, 1]$ edge weights is proved in \Cref{lemma:edge_weight_1}).
Then
\[
    \frac{1}{|S'|^2} \MaxCut(S') \in \frac{\MaxCut(X')}{|X'|^2} \pm \epsilon.
\]
Applying \Cref{lemma:relate_value}, and observe that $|X'| = N$,
the above equivalents to
\begin{align*}
    \frac{4\lambda^2 Q}{|S'|^2} \MaxCut(S')
    &\in (1 \pm \epsilon) \cdot \MaxCut(X) \pm \epsilon \cdot 4\lambda^2 Q \\
    &\in (1 \pm O(\lambda^2 \epsilon)) \cdot \MaxCut(X)
\end{align*}
where the last inequality follows from $\MaxCut(X) \geq \Omega(Q)$.
We finish the proof by rescaling $\epsilon$.

     \section{Streaming Implementations}
\label{sec:streaming}

\begin{theorem}[Streaming Euclidean $\MaxCut$]
\label{thm:streaming}
There is a randomized streaming algorithm that,
given $0 < \epsilon < 1/2$, $p \geq 1$, integers $\Delta, d \geq 1$,
and an input dataset $X \subseteq [\Delta]^d $ presented as a dynamic stream,
uses space $\poly(\epsilon^{-1} d \log \Delta)$
and reports an estimate $\eta>0$ that with high probability (at least ${2}/{3}$)
is a $(1 + \epsilon)$-approximation to $\MaxCut(X)$ in $\ell_p$.
\end{theorem}

Our algorithm employs importance sampling as formulated in \Cref{thm:sampling},
and thus needs (enough) samples from a distribution $\mathcal{D}$
that corresponds to the input $X \subseteq [\Delta]^d $
with small parameter $\lambda>0$.
Given these samples, the algorithm can estimate $\MaxCut(X)$
by a brute-force search on the (point-weighted) samples,
which can be done using small space.
Note that \Cref{thm:sampling} works for a general metric space,
hence it also applies to the $\ell_p$ case as we require.
We thus focus henceforth on performing importance sampling
from a dataset $X$ that is presented as a dynamic stream,
as formalized next in \Cref{lem:importance_sampling_algorithm}.

\begin{lemma}[Importance-Sampling Algorithm]
\label{lem:importance_sampling_algorithm}
There is a randomized streaming algorithm $\mathcal A$ that,
given $0 < \epsilon < 1/2$, $p \geq 1$, integers $\Delta, d \geq 1$,
and an input dataset $X \subseteq [\Delta]^d$ presented as a dynamic stream,
it uses space $\poly(\epsilon^{-1} d \log \Delta)$
and reports $z^* \in X\cup\{\perp\}$ together with $p^*\in[0,1]$.
The algorithm has a random initialization with success probability at least $0.99$,\footnote{It is convenient to separate the random coins of the algorithm into two groups,
  even though they can all be tossed before the stream starts.
  We refer to the coin tosses of the first group as an initialization step,
  and condition on their ``success'' when analyzing the second group of coins.
  The algorithm cannot tell whether its initialization was successful,
  and thus this event appears only in the analysis (in \Cref{lem:tree_embeddings_prob}).
}
and conditioned on a successful initialization, its random output satisfies:
(1) with probability at least $1 - 1 / \poly(\Delta^d)$,
\[
  \forall x \in X, \qquad
  \Pr[ z^* = x ] \ge \frac{1}{\lambda} \frac{q(x)}{Q},
\]
for $q(x) := \sum_{y \in X} \dist(x,y)$,
$Q := \sum_{x \in X} q(x)$, $\dist = \ell_p$, and
$\lambda := \poly(d \log \Delta)$;
and (2) whenever $z^*\neq \perp$,
\[
  z^*=x\in X \quad\Longrightarrow\quad
  p^* \in (1 \pm \epsilon) \cdot \Pr[z^* = x] .
\]
\end{lemma}

The somewhat intricate statement of \Cref{lem:importance_sampling_algorithm}
is very useful to generate many samples with a large success probability.
The obvious approach to generate $t$ samples
is to run $t$ executions of this algorithm (all in parallel on the same stream)
using independent coins, but then the success probability is only $0.99^t$.
Consider instead running $t$ parallel executions,
using the \emph{same initialization coins} but otherwise independent coins,
which requires total space $t\cdot \poly(\epsilon^{-1} d \log \Delta)$.
Then with probability at least $0.99$ the initialization succeeds,
in which case the $t$ executions produce $t$ independent samples,
each of the form $(z^*,p^*)$ and satisfies the two guarantees in the lemma.

\subsection{The Importance-Sampling Algorithm (Proof of \Cref{lem:importance_sampling_algorithm})}
\label{sec:proof_importance_sampling}

Our plan is to implement the importance sampling on a tree metric generated
by a randomized embedding of the input dataset.
The notion of randomized tree embedding was first proposed in~\cite{DBLP:conf/focs/Bartal96} for arbitrary metric spaces,
and the specific embedding that we employ was given by~\cite{Indyk04}
for $\ell_p$ metrics presented as a stream of points.
We describe this tree embedding below.
We stress that our algorithm can be easily implemented in low space
because it does not need to compute the entire embedding explicitly;
for instance, the algorithm's initialization picks random coins,
which determine the embedding but do not require any further computation.

\paragraph{Initialization Step: Randomized Tree Embedding~\cite{DBLP:conf/focs/Bartal96,CCGGP98,Indyk04}}
Assume without loss of generality that $\Delta \geq 1$ is an integral power of $2$,
and let $L := 1 + d \log \Delta$.
Let $\{\calG_i\}_{i=0}^{L}$ be a recursive partitioning of the grid $[2\Delta]^d$
into squares,\footnote{Strictly speaking, these squares are actually hypercubes (sometimes called cells or grids), but we call them squares for intuition.}
as follows.
Start with $\calG_0$ being a trivial partitioning that has one part corresponding to the entire grid $[2\Delta]^d$,
and for each $i \geq 0$, subdivide every square in $\calG_{i}$ into $2^d$ squares of half the side-length,
to obtain a partition $\calG_{i + 1}$ of the entire grid $[2\Delta]^d$.
Thus, every $\calG_{i}$ is a partition into squares of side-length $2^{i}$.
The recursive partitioning $\{\calG_i\}_i$ naturally defines a rooted tree $T$,
whose nodes are the squares inside all the $\calG_i$'s,
that if often called a \emph{quadtree decomposition}
(even though every tree node has $2^d$ children rather than $4$).
Finally, make the quadtree $T$ random
by shifting the entire recursive partitioning by a vector $-v_{\textrm{shift}}$,
where $v_{\textrm{shift}}$ is chosen uniformly at random from $[\Delta]^d$.
(This is equivalent to shifting the dataset $[\Delta]^d$ by $v_{\textrm{shift}}$,
which explains why we defined the recursive partitioning over an extended grid $[2\Delta]^d$.)
Every node in $T$ has a \emph{level} (or equivalently, depth),
where the root is at level $1$,
and the level of every other node is one bigger than that of its parent node.
The \emph{scale} of a tree node is the side-length of the corresponding square.
Observe that leaves of $T$ have scale $2^0$
and thus correspond to squares that contain a single grid point;
moreover, points $x\in [\Delta]^d$ correspond to distinct leaves in $T$.
Define the weight of an edge in $T$ between a node $u$ at scale $2^i$
and its parent as $d^{\frac{1}{p}} \cdot 2^i$ (i.e., the diameter of $u$'s square).
Define a \emph{tree embedding} of $[\Delta]^d$
by mapping every point $x \in [\Delta]^d$ to its corresponding leaf in $T$,
and let the \emph{tree distance} between two points $x, y \in [\Delta]^d$,
denoted $\dist_T(x, y)$,
be the distance in $T$ between their corresponding leaves.
The following lemma bounds the distortion of this randomized tree embedding.
We remark that a better distortion of $O(d^{\max\{\frac{1}{p}, 1 - \frac{1}{p}\}} \log \Delta)$ may be obtained
via a different technique that is less suitable for streaming~\cite{CCGGP98}.

\begin{lemma}[{\cite[Fact 1]{Indyk04}}]
  \label{lem:tree_embeddings}
  Let $T$ be a randomized tree as above.
  Then for all $x,y \in [\Delta]^d$,
  \begin{align*}
& \dist_T(x,y) \geq \dist(x,y) ;
    \\
\E[& \dist_T(x,y)] \leq O\left(d \log \Delta\right) \dist(x,y) .
  \end{align*}
\end{lemma}

\paragraph{Streaming Implementation of Randomized Tree Embedding}
We emphasize that in our definition of the quadtree $T$ is non-standard
as it contains the entire grid $[\Delta]^d$ as leaves
(the standard approach is to recursively partition only squares
that contain at least one point from the dataset $X$).
The advantage of our approach is that the tree is defined obliviously
of the dataset $X$ (e.g., of updates to $X$).
In particular, the leaf-to-root path from a point $x \in [\Delta]^d$
is well-defined regardless of $X$
and can be computed on-the-fly (without constructing the entire tree $T$)
using time and space $\poly(d \log \Delta)$,
providing sufficient information for evaluating the tree distance.

Our streaming algorithm samples such a tree $T$ as an initialization step,
i.e., before the stream starts,
which requires small space because it can be done implicitly
by picking $\poly(d \log \Delta)$ random bits that describe the random shift vector $v$.
Next, we show in \Cref{lem:tree_embeddings_prob}
that this initialization step succeeds with $0.99$ probability,
and on success, every distance $\dist(x,y)$ for $x,y\in X$
is well-approximated by its corresponding $\dist_T(x,y)$.
In this case, the sampling of points $x$ with probability proportional to $q(x)$
can be replaced by sampling with probabilities that are derived from the tree metric.
More specifically, the probability of sampling each $x\in X$
deviates from the desired probability $\frac{q(x)}{Q}$
by at most a factor of $\poly(d \log \Delta)$.
We remark that the event of success does depend on the input $X$,
but the algorithm does not need to know whether the initialization succeeded.

\begin{lemma}
  \label{lem:tree_embeddings_prob}
  For $x \in X$, let $q_T(x) := \sum_{y \in X} \dist_T(x,y)$
  and let $Q_T := \sum_{x \in X}q_T(x)$.
  Then
  \[
    \Pr_T\left[\forall x \in X,\  \frac{q_T(x)}{Q_T}
    \ge \frac{1}{O\left(d \log \Delta\right)}\frac{q(x)}{Q}\right] \geq 0.99.
  \]
\end{lemma}

\begin{proof}
  Fix some $x \in X$. By \Cref{lem:tree_embeddings},
  \begin{equation} \label{eqn:upper_bound_for_numerator}
    q_T(x) = \sum_{y \in X} \dist_T(x,y) \ge \sum_{y \in X} \dist(x,y)  = q(x)
  \end{equation}
  and
  \[
    \E\left[\sum_{y \in X}q_T(y)\right]
    = \E \left[\sum_{y,y' \in X} \dist_T(y,y')\right]
\le O\left(d \log \Delta\right)\sum_{y,y' \in X} \E [\dist(y,y')].
  \]
  By Markov's inequality, with high constant probability,
  \begin{equation}
    \label{eqn:sum_ub}
    \sum_{y \in X}q_T(y) \le
      O\left(d \log \Delta\right)\sum_{y,y' \in X} \E[\dist(y,y')].
  \end{equation}
  We finish the proof by combining \eqref{eqn:sum_ub} and \eqref{eqn:upper_bound_for_numerator}.
\end{proof}

\paragraph{Sampling w.r.t.\ Tree Distance}
In the remainder of the proof,
we assume that the random tree $T$ was already picked
and condition on its success as formulated in \Cref{lem:tree_embeddings_prob}.
This lemma shows that it actually suffices to sample each $x$
with probability proportional to $q_T(x)$.
Next, we provide in \Cref{fact:qt_x} a different formula for $q_T(x)$
that is based on $x$'s ancestors in the tree $T$,
namely, on counting how many data points (i.e., from $X$)
are contained in the squares that correspond to these ancestors.
To this end, we need to set up some basic notation regarding $X$ and $T$.

\paragraph{The Input $X$ in the Tree $T$}
Let $n := |X|$ be the number of input points at the end of the stream.
For a tree node $v \in T$, let $X(v) \subseteq X$ be the set of points from $X$
that are contained in the square corresponding to $v$.
For $x \in X$ and $i \geq 1$,
let $\anc_i(x)$ be the level-$i$ ancestor of $x$ in $T$
(recalling that $x$ corresponds to a leaf).
By definition, $\anc_{L + 1}(x) := x$.
For $0 \leq i \leq L$, let $\beta_i := d^{\frac{1}{p}} \cdot 2^{L + 1 - i}$,
which is the edge-length between a level-$i$ node $u$ and its parent
(since the scale of a level-$i$ node is $2^{L + 1 - i}$).
Due to the tree structure, we have the following representation of $q_T(x)$.
\begin{fact}
  \label{fact:qt_x}
  For every $x \in X$, we have
  $q_T(x) = 2 \sum_{i=0}^{L} \beta_i  \cdot (n - |X(\anc_i(x))|)$.
\end{fact}

For each level $i$, let $h_i$ be
a level-$i$ node whose corresponding square contains the most points from $X$,
breaking ties arbitrarily.
Next, we wish to identify a \emph{critical} level $k$;
ideally, this is the last level going down from the root, i.e., largest $i$,
such that $|X(h_i)| \ge 0.6n$ (the constant $0.6$ is somewhat arbitrary).
However, it is difficult to find this $k$ exactly in a streaming algorithm,
and thus we use instead a level $\tilde{k}$ that satisfies a relaxed guarantee
that only requires estimates on different $|X(h_i)|$,
as follows.
Let us fix henceforth two constants $0.5 < \sigma^- \leq \sigma^+ \leq 1$.

\begin{definition}[Critical Level]
Level $1 \leq \tilde{k} < L + 1$ is called \emph{$(\sigma^-, \sigma^+)$-critical},
if $|X(h_{\tilde{k}})| \geq \sigma^- n$ and $|X(h_{\tilde{k} + 1})| \leq \sigma^+ n$.
\end{definition}

Suppose henceforth that $\tilde{k}$ is a $(\sigma^-, \sigma^+)$-critical level.
(Such a critical level clearly exists, although its value need not be unique.)
Since $|X(h_i)| \geq |X(h_{i + 1})|$ for every $i<\tilde{k}$
(because $h_i$ contains the most points from $X$ at level $i$),
we know that $|X(h_i)| \geq \sigma^- n$ for every $i \leq \tilde{k}$ (not only for $i = \tilde{k}$),
and $|X(h_i)| \leq \sigma^+ n$ for every $i > \tilde{k}$.

\begin{fact}
  \label{fact:path}
  Each $h_{i}$ is the parent of $h_{i+1}$ for $1 \leq i \leq \tilde{k} - 1$,
  hence $(h_1, \ldots, h_{\tilde k})$ forms a path from the root of $T$.
\end{fact}

Next, we further ``simplify'' the representation of $q_T(x)$,
by introducing an approximate version of it that requires even less information about $x$.
Specifically, we introduce in \Cref{def:tilde_q}
a sequence of $O(L)$ values that are independent of $x$,
namely, one value $\tilde{q}_i$ for each level $i\leq \tilde{k}$,
and then we show in \Cref{lem:estimator} that for every $x \in X$,
we can approximate $q_T(x)$ by one of these $O(L)$ values,
namely, by $\tilde{q}_i$ for a suitable level $i=\ell(x)$.

\begin{definition}[Estimator for $q_T$]
  \label{def:tilde_q}
  For $1 \leq i \leq \tilde{k}$, define
  \[
    \tilde{q}_i := n\beta_i + \sum_{j \leq i} \beta_j \cdot (n - |X(h_j)|).
  \]
\end{definition}

\paragraph{Relating $q_T$ and $\tilde{q}$}
For $x\in X$, let $\ell(x)$ be the maximum level $1 \leq j \leq \tilde{k}$
such that $\anc_j(x) = h_j$.
This is well-defined, because $j = 1$ always satisfies that $\anc_j(x) = h_j$.
The next lemma shows that $q_T(x)$ can be approximated
by $\tilde{q}_i$ for $i = \ell(x)$.

\begin{lemma}
\label{lem:estimator}
Let $\tilde{k}$ be a $(\sigma^-, \sigma^+)$-critical level. Then
\[
  \forall x\in X,
  \qquad
  \tilde{q}_{\ell(x)}  = \Theta(1) \cdot q_T(x).
\]
\end{lemma}
\begin{proof}
  \begin{align}
    \frac{1}{2} q_T(x)
    &= \sum_{i=0}^{L}\beta_i \cdot (n - |X(\anc_i(x))|) \nonumber \\
    &= \sum_{i \leq \ell(x)} \beta_i \cdot (n - |X(h_i)|)
      + \sum_{i > \ell(x)} \beta_i \cdot (n - |X(\anc_i(x))|) \label{eqn:anc_to_h} \\
    &\in \sum_{i \leq \ell(x)} \beta_i \cdot (n - |X(h_i)|)
      + [\min\{ \sigma^-, 1 - \sigma^+ \}, 1] \cdot n \sum_{i > \ell(x)}\beta_i \label{eqn:use_mu} \\
    &\in  \sum_{i \leq \ell(x)} \beta_i\cdot (n - |X(h_i)|)
      + [\min\{ \sigma^-, 1 - \sigma^+ \}, 1] \cdot n \beta_{\ell(x)} \nonumber \\
    &\in [\min\{ \sigma^-, 1 - \sigma^+ \}, 1] \cdot \tilde{q}_{\ell(x)}. \nonumber
  \end{align}
  In the above, \eqref{eqn:anc_to_h} follows from the fact that $\anc_i(x) = h_i$ for $i \leq \ell(x)$ (by the definition of $\ell(x)$ and the property that $(h_1, \ldots, h_{\tilde k})$ forms a path from \Cref{fact:path}).
  \eqref{eqn:use_mu} follows from the definition of $(\sigma, \mu)$-critical
  and the definition of $\ell$.
\end{proof}

The next lemma shows that the sequence $\tilde{q}_1,\ldots,\tilde{q}_{\tilde k}$
is non-increasing.
\begin{fact}
  \label{fact:qt_dec}
  $\tilde{q}_1 = \beta_1 n$,
  and for every $2 \leq i \leq \tilde k$, we have $\tilde{q}_{i} \leq \tilde{q}_{i - 1}$.
\end{fact}
\begin{proof}
  The fact for $i = 1$ is immediate. Now consider $i \geq 2$. We have
  \[
    \tilde{q}_{i - 1} - \tilde{q}_{i}
    = n(\beta_{i - 1} - \beta_i) - \beta_i \cdot (n - |X(h_i)| )
    = \beta_i \cdot |X(h_i)|
    \geq 0,
  \]
  which verifies the lemma.
\end{proof}

\paragraph{Alternative Sampling Procedure}
Recall that level $\tilde k$ is assumed to be $(\sigma^-, \sigma^+)$-critical
for fixed constants $0.5 < \sigma^- \leq \sigma^+ \leq 1$.
We plan to sample $x \in X$ with probability proportional to $\tilde{q}_{\ell(x)}$,
and by \Cref{lem:estimator} this only loses an $O(1)$ factor in the bound $\lambda$ needed for importance sampling (as in \Cref{lem:importance_sampling_algorithm}).
For $1 \leq i \leq \tilde k$, define $X_i := \{ x \in X \mid \ell(x) = i \}$.
Notice that $\{X_i\}_{i=1}^{\tilde k}$ forms a partition of $X$, and
\begin{equation}
  \label{eqn:Xi}
  X_i =
  \begin{cases}
    X(h_{i}) \setminus X(h_{i+1}) & \text{if $1 \leq i \leq \tilde{k} - 1$;} \\
    X(h_{\tilde k}) & \text{if $i = \tilde{k}$.}
  \end{cases}
\end{equation}
By definition, points in the same $X_i$ have the same $\tilde{q}_{\ell(x)}$,
and thus also the same sampling probability.
A natural approach to sampling a point from $X$ with the desired probabilities
is to first pick a random $i\in[\tilde{k}]$ (non-uniformly)
and then sample uniformly a point from that $X_i$.
But unfortunately, it is impossible to sample uniformly from $X_i$ in streaming
(this is justified in \Cref{claim:lb_sampling_light}),
and thus we shall sample instead from an ``extended'' set $\Xext_i\supseteq X_i$,
defined as follows.
\begin{equation}
  \label{eqn:Xext}
  \Xext_i :=
  \begin{cases}
    X \setminus X(h_{i+1}) & \text{if $1 \leq i \leq \tilde{k} - 1$;} \\
    X(h_{\tilde k}) & \text{if $i = \tilde{k}$.}
  \end{cases}
\end{equation}
The path structure of $\{h_i\}_{i}$ (\Cref{fact:path}) implies the following.

\begin{fact}
  \label{fact:xext}
  For every $1 \leq i < \tilde{k}$, we have
  $\Xext_i = X_1 \cup \ldots \cup X_i$.
\end{fact}

We describe in \Cref{alg:offline_sampling} a procedure
for sampling $x \in X$ with probability proportional to $\tilde{q}_{\ell(x)}$,
based on the above approach
of picking a random $i\in[\tilde{k}]$ (from a suitable distribution)
and then sampling uniformly a point from that $\Xext_i$.
We then prove in \Cref{lem:prob} that this procedure samples from $X$
with probabilities proportional to $\tilde{q}_{\ell(x)}$, up to an $O(L)$ factor.

\begin{remark}
\label{rem:ChangeOrderOfSum}
Sampling from the extended sets ($\Xext_i$ instead of $X_i$)
can significantly bias the sampling probabilities,
because the ``contribution'' of a point $x\in X$
can increase by an unbounded factor.
On the one hand, this can increase the sampling probability of that $x$,
which is not a problem at all.
On the other hand, it might increase the total contribution of all points
(and thus decrease some individual sampling probabilities),
but our analysis shows that this effect is bounded by an $O(L)$ factor.
The intuition here is that
$q(x)$ represents the sum of distances from $x$ to all other points $y\in X$,
and we can rearrange their total $\sum_x q(x)$ by the ``other'' point $y\in X$,
and the crux now is that the contribution of each $y\in X$
increases by at most $O(L)$ factor.
\end{remark}

\begin{algorithm}[ht]
  \caption{Alternative sampling procedure (offline)}
  \label{alg:offline_sampling}
  \begin{algorithmic}[1]
    \State draw a random $i^*$ where each $1 \leq i \leq \tilde{k}$ is picked with probability
    $r_i := \frac{|\Xext_i| \tilde{q}_i}{\sum_{j=1}^{\tilde k} |\Xext_j| \tilde{q}_j}$ \label{line:offline_i}
    \State draw $x \in \Xext_i$ uniformly at random \label{line:offline_x}
    \State return $z^* = x$ as the sample, together with $p^* = \sum_{i = \ell(x)}^{\tilde{k}} \frac{r_i}{|\Xext_i|} $ as its sampling probability
    \label{line:offline_ret}
  \end{algorithmic}
\end{algorithm}

\begin{lemma}
  \label{lem:prob}
  \Cref{alg:offline_sampling} samples every $x \in X$ with probability
  $\Pr[z^* = x] = \sum_{i = \ell(x)}^{\tilde{k}} \frac{r_i}{|\Xext_i|}$,
  exactly as line 3 reports in $p^*$, and furthermore this is bounded by
  $\Pr[z^* = x]
  \geq \Omega\left(\frac{1}{L}\right) \frac{\tilde{q}_{\ell(x)}}
    { \sum_{x \in X} \tilde{q}_{\ell(x)} }$.
\end{lemma}

\begin{proof}
  Observe that $x \in X_{\ell(x)}$,
  and by \Cref{fact:xext}, this point $x$ can only be sampled
  for $i \geq \ell(x)$. Therefore,
  $\Pr[z^* = x] = \sum_{i=\ell(x)}^{\tilde{k}} \frac{r_i}{|\Xext_i|} = p^*$.
  We bound this probability by
  \begin{equation}
    \label{eqn:prob}
    \Pr[z^* = x]
    = \sum_{i \geq \ell(x)} \frac{r_i}{|\Xext_i|}
    = \sum_{i \geq \ell(x)} \frac{\tilde{q}_i}
      { \sum_{j=1}^{\tilde k}{|\Xext_j|\tilde{q}_j} }
    = \frac{\sum_{i \geq \ell(x)} \tilde{q}_i}
      { \sum_{j=1}^{\tilde k}{|\Xext_j|\tilde{q}_j} }
    \geq \frac{\tilde{q}_{\ell(x)}}
      { \sum_{j=1}^{\tilde k}{|\Xext_j|\tilde{q}_j} }.
  \end{equation}
  Next, to bound the denominator ${ \sum_{j=1}^{\tilde k}{|\Xext_j|\tilde{q}_j} }$,
  observe that $|\Xext_j| =  \sum_{i = 1}^{j} |X_i|$ for all $j<\tilde{k}$
  (by \Cref{fact:xext}),
  and therefore
  \begin{align*}
    \sum_{j=1}^{\tilde k}{|\Xext_j|\tilde{q}_j}
    &= |X_{\tilde k}|\tilde{q}_{\tilde k}
      + \sum_{j = 1}^{\tilde{k} - 1} \sum_{i = 1}^{j}|X_i| \tilde{q}_j
      = |X_{\tilde k}|\tilde{q}_{\tilde k}
      + \sum_{i = 1}^{\tilde{k} - 1}\sum_{j = i}^{\tilde{k} - 1} |X_i| \tilde{q}_j
    \leq |X_{\tilde k}|\tilde{q}_{\tilde k} + \tilde{k} \cdot \sum_{i = 1}^{\tilde{k} - 1}|X_i| \tilde{q}_i \\
    &\leq (L + 1) \sum_{i=1}^{\tilde k}{ |X_i| \tilde{q}_i }
= (L + 1) \sum_{x \in X}{\tilde{q}_{\ell(x)}},
  \end{align*}
  where the first inequality is by the monotonicity of $\tilde{q}_i$'s (\Cref{fact:qt_dec}).
  Combining this with \eqref{eqn:prob}, the lemma follows.
\end{proof}

\paragraph{Implementing \Cref{alg:offline_sampling} in Streaming}
To implement \Cref{alg:offline_sampling} in streaming,
we first need a streaming algorithm that finds a critical level $\tilde{k}$
using space $O(\poly(d \log \Delta))$.
We discuss this next.

\paragraph{Finding $\tilde{k}$}
For each level $i$, we draw $\poly(d \log \Delta)$ samples $S_i \subseteq X$ uniformly at random from $X$.
We then count the number of samples that lie in each tree node (square) at level $i$,
and let $m_i$ be the maximum count.
We let $\tilde{k}$ be the largest level $i$
such that $\frac{m_i}{|S_i|} \geq 0.6$.
By a standard application of Chernoff bound,
with probability at least $1 - \poly(\Delta^d)$,
this level $\tilde{k}$ is $(0.55, 0.65)$-critical.
Moreover, this process can be implemented in streaming using space $\poly(d \log \Delta)$,
by maintaining, for each level $i$,
only $|S_i|=\poly(d \log \Delta)$ independent $\ell_0$-samplers (\Cref{lem:l0sampler}) on the domain $[\Delta]^d$.
A similar approach can be used
to $(1+\epsilon)$-approximate the size of $X(h_i)$ for every $i \leq \tilde{k}$,
and also sample uniformly from these sets,
using space $O(\poly(\epsilon^{-1}d\log \Delta))$
and with failure probability $1 - 1 / \poly(\Delta^d)$
(by using \Cref{lem:l0estimator,lem:l0sampler}).

\paragraph{Estimating and Sampling from $X \setminus X(h_i)$}
We also need to estimate $\tilde{q}_i$ and $|\Xext_i|$,
and to sample uniformly at random from $\Xext_i$, for every $i \leq \tilde{k}$.
The case $i = \tilde{k}$ was already discussed,
because $\Xext_{\tilde k} = X(h_{\tilde k})$.
It remains to consider $i < \tilde{k}$, in which case we need
to $(1 \pm \epsilon)$-approximate the size of $X \setminus X(h_i)$,
and also to sample uniformly at random from that set,
and we can assume that $|X(h_i)| > 0.5 n$.
We provide such a streaming algorithm in \Cref{lemma:streaming_light} below,
which we prove in \Cref{sec:proof_streaming_light}.
This lemma is stated in a more general form
that may be of independent interest,
where the input is a frequency vector $x\in \mathbb{R}^N$
(i.e., a stream of insertions and deletions of items from domain $[N]$)
and access to a function $\mathcal{P}:[N] \to [N']$, for $N'\leq N$,
that can be viewed as a partition of the domain into $N'$ parts.
In our intended application, the domain $[N]$ will be the grid $[\Delta]^d$,
and the partition $\mathcal{P}$ will be its partition into squares of a given level $i$;
observe that it is easy to implement $\mathcal{P}$
as a function that maps each grid point to its level-$i$ square.
Roughly speaking, the streaming algorithm in \Cref{lemma:streaming_light}
samples uniformly from the support set $\supp(x) = \{ i\in[N]: x_i\neq 0\}$,
but excluding indices that lie in the part of $\mathcal{P}$ that is heaviest,
i.e., has most nonzero indices, assuming it is sufficiently heavy.
In our intended application, this method samples uniformly from the input
$X\subset[\Delta]^d$ but excluding points that lie in the heaviest square,
i.e., uniformly from $X \setminus X(h_i)$.
(square with the largest number of input points).

\begin{lemma}[Sampling from Light Parts]
  \label{lemma:streaming_light}
  There exists a streaming algorithm,
  that given $0 < \epsilon, \delta, \sigma < 0.5$,
  integers $N, N', M \geq 1$,
  a mapping $\mathcal{P} : [N] \to [N']$,
  and a frequency vector $x \in [-M, M]^N$ that is presented as a stream of additive updates,
  uses space $O(\poly(\epsilon^{-1}\sigma^{-1}\log(\delta^{-1}MN)))$,
  and reports a sample $i^* \in [N] \cup \{\nil\}$ and a value $r^* \geq 0$.
Let $X := \{ i \in [N] \mid x_i \neq 0\}$ be the support of $x$,
  and let $j_{\max} := \arg\max_{j \in [N']} |\mathcal{P}^{-1}(j) \cap X|$
  be the heaviest $\mathcal{P}$ with respect to $X$.
  If $\Xheavy := \mathcal{P}^{-1}(j_{\max}) \cap X$ satisfies
  $|\Xheavy| \geq (0.5 + \sigma) |X|$,
then with probability at least $1 - \delta$,
  \begin{itemize}
    \item $r^* \in (1 \pm \epsilon) \cdot |\Xlight|$ where $\Xlight := X \setminus \Xheavy$, and
    \item unless $\Xlight$ is empty,
      $i^* \in \Xlight$
      and moreover for all $i \in \Xlight$,
      it holds that $\Pr[i^* = i] = \frac{1}{|\Xlight|}$ (provided that $|\Xlight| \neq 0$).
\end{itemize}
\end{lemma}

In our application,
we will apply \Cref{lemma:streaming_light} in parallel for every level $i$,
with $N = \Delta^d$,
i.e., the items being inserted and deleted are points in $[\Delta]^d$,
and a mapping $\mathcal{P}$ defined by the level-$i$ squares (tree nodes),
i.e., for $x \in [\Delta]^d$ we define $\mathcal{P}(x)$ as the level-$i$ node that contains $x$.
We will set the failure probability to be $\delta = 1 / \poly(\Delta^{d})$
and a fixed $\sigma = 0.05$.
This way, conditioning on the success of \Cref{lemma:streaming_light},
we can compute $\tilde{q}_i$, $|\Xext_i|$ with error $(1 \pm \epsilon)$,
and sampled from $\Xext_i$ uniformly.

\paragraph{Concluding \Cref{lem:importance_sampling_algorithm}}
  In conclusion, our streaming algorithm initializes
  with sampling a randomly-shifted quadtree $T$ which defines a tree embedding, all in an implicit way.
  Then, assume $T$ is obtained and condition on the success of it, specifically \Cref{lem:tree_embeddings_prob} (with probability $0.99$),
  we use the streaming implementation of \Cref{alg:offline_sampling},
  as outlined above.
  The resultant $z^*$ and $p^*$ are the return value.
  The error bound on $z^*$ and $p^*$
  and the bound of $\lambda = O(\poly(d \log \Delta))$
  follow by \Cref{lem:tree_embeddings_prob} and \Cref{lem:prob},
  plus an additional error and failure probability
  introduced by streaming, which is bounded in the previous paragraphs.
  This finishes the proof.

\subsection{Sampling from The Light Parts (Proof of \Cref{lemma:streaming_light})}
\label{sec:proof_streaming_light}
\paragraph{An Offline Algorithm}
  Notice that $\{ \mathcal{P}^{-1}(y) \}_{y \in [N']}$ defines a partition of $[N]$.
  In our proof, we interpret $\mathcal{P} = \{ P_i \}_i$ as such a partition.
  Let $P_{\max} := \mathcal{P}^{-1}(y_{\max})$ be the part of $\mathcal{P}$ that contains the most from $X$, so $\Xheavy = P_{\max} \cap X$.
  We start with an offline algorithm, summarized in \Cref{alg:sampling_light}.
  \begin{algorithm}[ht]
    \caption{Sampling and estimating from the light part (offline)}
    \label{alg:sampling_light}
    \begin{algorithmic}[1]
      \State let $u \gets 2$,
      $s \gets \Theta(\log(N\delta^{-1}))$
      \State let $\mathcal{H} \gets \{ h_1, \ldots, h_s \}$ be a collection of independent random hash functions,
      where each $h \in \mathcal{H}$ ($h : \mathcal{P} \to [u]$) satisfies $\forall P \neq P'$, $\Pr[h(P) = h(P')] \leq 1 / u$ \label{line:hash}
      \For{$t \in [s]$}
        \State for $j \in [u]$, let $B_j \gets \left(\bigcup_{P \in \mathcal{P} : h_t(P) = j}P\right) \cap X$ \label{line:Bj}
        \State let $j^* \gets \arg\max_{j} |B_j|$ \label{line:jstar}
        \State let $D_t \gets X \setminus \bigcup_{P \in \mathcal{P} : h_t(P) = j^*}{P}$ \label{line:Di}
      \EndFor
      \State compute $\Dall \gets \bigcup_{t \in [s]} D_t$ \label{line:Dmax}
      \State return a uniform sample $i^* \in \Dall$, and report $r^* := |\Dall|$
      as the estimate for $|\Xlight|$
    \end{algorithmic}
  \end{algorithm}
  In the algorithm, we consider a set of
  $s = \Theta(\log(N \delta^{-1}))$ random hash functions $h_1, \ldots, h_s$
  that randomly map each part in $\mathcal{P}$ to one of
  $u = 2$ buckets (as in line~\ref{line:hash}).

  Then, consider some $h_t$ for $t \in [s]$.
  Let $B_j$ ($j \in [u]$) be the elements from all parts that are mapped by $h_t$ to the bucket $j$  (in line~\ref{line:Bj}).
  We find $j^*$ as the bucket that contains the most elements from $X$ (in line~\ref{line:jstar}).
  Since we assume $|\Xheavy| \geq (0.5 + \sigma) |X| > 0.5 |X|$,
  we know the bucket $h_t(P_{\max})$ contains \emph{more} than $0.5 |X|$ elements form $X$ (recalling that $P_{\max} = \mathcal{P}^{-1}(y_{\max})$ is the part that contains the most from $X$),
  and this implies $h_t(P_{\max})$ must be the bucket that contains the most elements from $X$.
  Hence,
  \begin{equation}
    \label{eqn:jstart_max}
    j^* = h_t(P_{\max}).
  \end{equation}
Next, we drop the elements that lie in the bucket $j^*$,
  and take the remaining elements, as $D_t$ (in line~\ref{line:Di}).
  While $D_t$ certainly does not contain any element from $\Xheavy$
  (by \eqref{eqn:jstart_max} and the definition of $D_t$'s),
  $D_t$ is only a subset of $\Xlight$.
Hence, we take the union of all $D_t$'s (over $t \in [s]$),
  denoted as $\Dall$ (in line~\ref{line:Dmax}), which equals $\Xlight$ with high probability.

  \paragraph{Analysis of $\Dall$}
  For every $i \in \Xlight$, every $t \in [s]$,
  \[
    \Pr[i \notin D_t]
    = \Pr[h_t(P_i) = h_t(P_{\max}) ]
    \leq \frac{1}{u} = \frac{1}{2},
  \]
  where $P_i \in \mathcal{P}$ is the part that $i$ belongs to.
  Therefore, by the independence of $h_t$'s, we know for every $i \in \Xlight$,
  \[
    \Pr[i \notin \Dall]
    = \Pr[\forall t \in [s], i \notin D_t]
    \leq \frac{1}{2^s} = \frac{\delta}{\poly(N)}.
  \]
  Taking a union bound over $i \in \Xlight$, we have
  \[
    \Pr[\exists i \in \Xlight, i \notin \Dall]
    \leq \frac{\delta}{\poly(N)} |\Xlight|
    \leq \delta.
  \]
  Hence, we conclude that
  \[
    \Pr[\Dall = \Xlight] \geq 1 - \delta.
  \]
  Conditioning on the event that $\Dall = \Xlight$,
  we conclude that $r^* = |\Dall| = |\Xlight|$,
  and that $\forall i \in \Xlight$,
  $\Pr[i^* = i] = \frac{1}{|\Dall|} = \frac{1}{|\Xlight|}$.

  \paragraph{Streaming Algorithm}
  It remains to give a streaming implementation for \Cref{alg:sampling_light}.
  Before the stream starts, we initialize several streaming data structures.
  We start with building the hash functions $\mathcal{H}$,
  and this can be implemented using space $\poly(\log N)$,
  by using hash families of limited independence.
  Next, we maintain for every $t \in [s]$, for every bucket $j \in [u]$,
  an $\ell_0$-sampler $\mathcal{L}^{(t)}_j$ (\Cref{lem:l0sampler}) with failure probability $O(\frac{\delta}{us})$,
  as well as an $\ell_0$-norm estimator $\mathcal{K}^{(t)}_j$
  (\Cref{lem:l0estimator}) with failure probability $O(\frac{\delta}{us})$
  and error guarantee $\epsilon \sigma \leq \min\{\epsilon, \sigma\}$,
  both on domain $[n]$.
  The setup of the failure probabilities immediately implies that
  with probability $1 - \delta$, all data structures succeed simultaneously,
  and we condition on their success in the following argument.
  Since we need to combine the linear sketches $\mathcal{L}^{(t)}_j$'s in later steps,
  for every $t \in [s]$ and $j \in [u]$, we use the same random seeds among all $\ell_0$-samplers $\{\mathcal{L}^{(t)}_j\}$'s,
  so that they can be ``combined'' by simply adding up.
  Also do the same for $\mathcal{K}^{(t)}_j$'s.
  Another detail is that, strictly speaking,
  we need $O(1)$ independent ``copies'' of every $\mathcal{K}$ and $\mathcal{L}$,
  since we need to query each of them $O(1)$ times.
  As this only enlarges the space by an $O(1)$ factor, we omit this detail for the sake of presentation.

  Whenever an update for element $i \in [n]$ is received,
  we update $\mathcal{L}^{(t)}_{j_i}$ and $\mathcal{K}^{(t)}_{j_i}$ for every $t \in [s]$,
  where $j_i := h_t(P_i)$, and $P_i \in \mathcal{P}$
  is the unique part that contains $i$.

  When the stream terminates, we proceed to generate the sample $i^* \in \Xlight$ and the estimate $r^*$ for $\Xlight$.
  For $t \in [s]$, $j \in [u]$, query $\mathcal{K}^{(t)}_j$
  to obtain an estimator for $|B_j|$ (line~\ref{line:Bj})
  within $(1 \pm \epsilon\sigma)$ factor.
  Use these estimations to find $j^*$ (line~\ref{line:jstar}).
  Note that this $j^*$ is the same as computing using exact $|B_j|$ values.
  To see this, the key observation is that, $|\Xheavy| \geq ( 0.5 + \sigma) |X|$,
  while for every $P \in \mathcal{P} \setminus P_{\max}$ we have $|P| \leq (0.5 - \sigma) |X|$.
  Hence, to precisely find $j^*$, it suffices to distinguish between
  subsets $P$, $P'$ such that
  $|P| \geq (0.5 + \sigma) |X|$ and $|P'| \leq (0.5 - \sigma)|X|$.
  Even with a $(1 \pm \epsilon \sigma)$ error (which is the error of our $\mathcal{K}$'s),
  this gap is still $\frac{0.5 + \sigma}{ 0.5 - \sigma} \cdot \frac{1 + \epsilon \sigma}{ 1 - \epsilon \sigma} > 1$ which is large enough.

  Next, compute $\mathcal{L}^{(t)} := \sum_{j \in [u] \setminus \{j^*\}}\mathcal{L}^{(t)}_j$
  as the $\ell_0$-sampler that corresponds to $D_t$ (line~\ref{line:Di}),
  and obtain $\mathcal{K}^{(t)} := \sum_{j \in [u] \setminus \{j^*\}}\mathcal{K}^{(t)}_j$ similarly.
  We can do this since we use the same random seeds among $\mathcal{L}^{(t)}_j$'s (and the same has been done to $\mathcal{K}$).
  We further compute $\mathcal{L} := \sum_{t \in [s]} \mathcal{L}^{(t)}$
  whose support corresponds to $\Dall$.
  Define $\mathcal{K} := \sum_{t \in [s]} \mathcal{K}^{(t)} $ similarly.
  The final return values $i^*$ and $r^*$ are given by querying $\mathcal{L}$ and $\mathcal{K}$.
  Note that on the success of the $\ell_0$-sampler $\mathcal{L}$,
  the probability for $i^* = i$ for each $i \in W$ is exactly $\frac{1}{|W|}$ (\Cref{lem:l0sampler}).
  However, the $r^*$ deviates from $|W|$ by a multiplicative $(1 \pm \epsilon)$ factor.

  In conclusion, the analysis of \Cref{alg:sampling_light} still goes through by using the estimated values as in the above procedure,
  except that one needs to rescale $\epsilon$ and $\delta$ by a constant factor,
  to compensate the error and failure probability introduced by the streaming data structures.
  This finishes the proof of \Cref{lemma:streaming_light}.

    \subsection*{Acknowledgments}
    We thank Christian Sohler for pointing us to the dimension reduction result
    in~\cite{DBLP:conf/wads/LammersenSS09, lammersen2011approximation}.
    We also thank an anonymous reviewer for pointing out how to simplify
    our proof of this result (\Cref{thm:maxcut_jl}).

    \bibliography{ref.bib}

\newcommand{\etalchar}[1]{$^{#1}$}
\begin{thebibliography}{AdlVKK03}

\bibitem[ABIW09]{DBLP:conf/focs/AndoniBIW09}
Alexandr Andoni, Khanh~Do Ba, Piotr Indyk, and David~P. Woodruff.
\newblock Efficient sketches for earth-mover distance, with applications.
\newblock In {\em {FOCS}}, pages 324--330. {IEEE} Computer Society, 2009.

\bibitem[AdlVKK03]{DBLP:journals/jcss/AlonVKK03}
Noga Alon, Wenceslas~Fernandez de~la Vega, Ravi Kannan, and Marek Karpinski.
\newblock Random sampling and approximation of {MAX-CSPs}.
\newblock {\em J. Comput. Syst. Sci.}, 67(2):212--243, 2003.

\bibitem[AG09]{DBLP:conf/icalp/AhnG09}
Kook~Jin Ahn and Sudipto Guha.
\newblock Graph sparsification in the semi-streaming model.
\newblock In {\em {ICALP} {(2)}}, volume 5556 of {\em Lecture Notes in Computer
  Science}, pages 328--338. Springer, 2009.

\bibitem[AHV04]{AHV04}
Pankaj~K. Agarwal, Sariel {Har-Peled}, and Kasturi~R. Varadarajan.
\newblock Approximating extent measures of points.
\newblock {\em J. ACM}, 51(4):606--635, 2004.
\newblock \href {https://doi.org/10.1145/1008731.1008736}
  {\path{doi:10.1145/1008731.1008736}}.

\bibitem[Aro98]{Arora98}
Sanjeev Arora.
\newblock Polynomial time approximation schemes for {Euclidean} traveling
  salesman and other geometric problems.
\newblock {\em Journal of the ACM}, 45(5):753--782, 1998.
\newblock \href {https://doi.org/10.1145/290179.290180}
  {\path{doi:10.1145/290179.290180}}.

\bibitem[AS15]{AS15}
Pankaj~K. Agarwal and R.~Sharathkumar.
\newblock Streaming algorithms for extent problems in high dimensions.
\newblock {\em Algorithmica}, 72(1):83--98, 2015.
\newblock \href {https://doi.org/10.1007/s00453-013-9846-4}
  {\path{doi:10.1007/s00453-013-9846-4}}.

\bibitem[Bar96]{DBLP:conf/focs/Bartal96}
Yair Bartal.
\newblock Probabilistic approximations of metric spaces and its algorithmic
  applications.
\newblock In {\em {FOCS}}, pages 184--193. {IEEE} Computer Society, 1996.

\bibitem[BFL{\etalchar{+}}17]{BFLSY17}
Vladimir Braverman, Gereon Frahling, Harry Lang, Christian Sohler, and Lin~F.
  Yang.
\newblock Clustering high dimensional dynamic data streams.
\newblock In {\em {ICML}}, volume~70 of {\em Proceedings of Machine Learning
  Research}, pages 576--585. {PMLR}, 2017.

\bibitem[CCG{\etalchar{+}}98]{CCGGP98}
Moses Charikar, Chandra Chekuri, Ashish Goel, Sudipto Guha, and Serge~A.
  Plotkin.
\newblock Approximating a finite metric by a small number of tree metrics.
\newblock In {\em {FOCS}}, pages 379--388. {IEEE} Computer Society, 1998.

\bibitem[CJK{\etalchar{+}}22]{CJKVW22}
Artur Czumaj, Shaofeng~H.{-}C. Jiang, Robert Krauthgamer, Pavel Vesel{\'{y}},
  and Mingwei Yang.
\newblock Streaming facility location in high dimension via new geometric
  hashing.
\newblock In {\em {FOCS}}. {IEEE}, 2022.

\bibitem[CJKV22]{DBLP:conf/icalp/CzumajJKV22}
Artur Czumaj, Shaofeng~H.{-}C. Jiang, Robert Krauthgamer, and Pavel
  Vesel{\'{y}}.
\newblock Streaming algorithms for geometric {Steiner} forest.
\newblock In {\em {ICALP}}, volume 229 of {\em LIPIcs}, pages 47:1--47:20.
  Schloss Dagstuhl - Leibniz-Zentrum f{\"{u}}r Informatik, 2022.
\newblock \href {https://doi.org/10.4230/LIPIcs.ICALP.2022.47}
  {\path{doi:10.4230/LIPIcs.ICALP.2022.47}}.

\bibitem[CJLW22]{DBLP:conf/stoc/ChenJLW22}
Xi~Chen, Rajesh Jayaram, Amit Levi, and Erik Waingarten.
\newblock New streaming algorithms for high dimensional {EMD} and {MST}.
\newblock In {\em {STOC}}, pages 222--233. {ACM}, 2022.

\bibitem[CLMS13]{DBLP:conf/soda/CzumajLMS13}
Artur Czumaj, Christiane Lammersen, Morteza Monemizadeh, and Christian Sohler.
\newblock ($1 + \varepsilon$)-approximation for facility location in data
  streams.
\newblock In {\em {SODA}}, pages 1710--1728. {SIAM}, 2013.

\bibitem[dlVK00]{DBLP:journals/rsa/VegaK00}
Wenceslas~Fernandez de~la Vega and Marek Karpinski.
\newblock Polynomial time approximation of dense weighted instances of
  {MAX-CUT}.
\newblock {\em Random Struct. Algorithms}, 16(4):314--332, 2000.

\bibitem[dlVK01]{VK01}
Wenceslas~Fernandez de~la Vega and Claire Kenyon.
\newblock A randomized approximation scheme for metric {MAX-CUT}.
\newblock {\em J. Comput. Syst. Sci.}, 63(4):531–541, dec 2001.
\newblock \href {https://doi.org/10.1006/jcss.2001.1772}
  {\path{doi:10.1006/jcss.2001.1772}}.

\bibitem[FIS08]{DBLP:journals/ijcga/FrahlingIS08}
Gereon Frahling, Piotr Indyk, and Christian Sohler.
\newblock Sampling in dynamic data streams and applications.
\newblock {\em Int. J. Comput. Geom. Appl.}, 18(1/2):3--28, 2008.

\bibitem[FL11]{DBLP:conf/stoc/FeldmanL11}
Dan Feldman and Michael Langberg.
\newblock A unified framework for approximating and clustering data.
\newblock In {\em {STOC}}, pages 569--578. {ACM}, 2011.

\bibitem[FS05]{DBLP:conf/stoc/FrahlingS05}
Gereon Frahling and Christian Sohler.
\newblock Coresets in dynamic geometric data streams.
\newblock In {\em {STOC}}, pages 209--217. {ACM}, 2005.

\bibitem[FSS20]{DBLP:journals/siamcomp/FeldmanSS20}
Dan Feldman, Melanie Schmidt, and Christian Sohler.
\newblock Turning big data into tiny data: Constant-size coresets for
  $k$-means, {PCA}, and projective clustering.
\newblock {\em {SIAM} J. Comput.}, 49(3):601--657, 2020.

\bibitem[GGR96]{GGR96}
Oded Goldreich, Shafi Goldwasser, and Dana Ron.
\newblock Property testing and its connection to learning and approximation.
\newblock In {\em {FOCS}}, pages 339--348. {IEEE} Computer Society, 1996.

\bibitem[GW95]{DBLP:journals/jacm/GoemansW95}
Michel~X. Goemans and David~P. Williamson.
\newblock Improved approximation algorithms for maximum cut and satisfiability
  problems using semidefinite programming.
\newblock {\em J. {ACM}}, 42(6):1115--1145, 1995.

\bibitem[HM04]{DBLP:conf/stoc/Har-PeledM04}
Sariel Har{-}Peled and Soham Mazumdar.
\newblock On coresets for $k$-means and $k$-median clustering.
\newblock In {\em {STOC}}, pages 291--300. {ACM}, 2004.

\bibitem[HSYZ18]{DBLP:journals/corr/abs-1802-00459}
Wei Hu, Zhao Song, Lin~F. Yang, and Peilin Zhong.
\newblock Nearly optimal dynamic $k$-means clustering for high-dimensional
  data.
\newblock {\em CoRR}, abs/1802.00459, 2018.

\bibitem[Ind04]{Indyk04}
Piotr Indyk.
\newblock Algorithms for dynamic geometric problems over data streams.
\newblock In {\em {STOC}}, pages 373--380. {ACM}, 2004.

\bibitem[JKS08]{DBLP:journals/toc/JayramKS08}
T.~S. Jayram, Ravi Kumar, and D.~Sivakumar.
\newblock The one-way communication complexity of hamming distance.
\newblock {\em Theory Comput.}, 4(1):129--135, 2008.

\bibitem[JL84]{JL84}
W.~B. Johnson and J.~Lindenstrauss.
\newblock Extensions of {L}ipschitz mappings into a {H}ilbert space.
\newblock In {\em Conference in modern analysis and probability (New Haven,
  Conn., 1982)}, pages 189--206. Amer. Math. Soc., 1984.

\bibitem[JST11]{JST11}
Hossein Jowhari, Mert Saglam, and G{\'{a}}bor Tardos.
\newblock Tight bounds for ${L}_p$ samplers, finding duplicates in streams, and
  related problems.
\newblock In {\em {PODS}}, pages 49--58. {ACM}, 2011.

\bibitem[JW21]{JW21}
Rajesh Jayaram and David Woodruff.
\newblock Perfect {$L_p$} sampling in a data stream.
\newblock {\em SIAM J. Comput.}, 50(2):382–439, 2021.
\newblock \href {https://doi.org/10.1137/18M1229912}
  {\path{doi:10.1137/18M1229912}}.

\bibitem[KK19]{DBLP:conf/stoc/KapralovK19}
Michael Kapralov and Dmitry Krachun.
\newblock An optimal space lower bound for approximating {MAX-CUT}.
\newblock In {\em {STOC}}, pages 277--288. {ACM}, 2019.

\bibitem[KKMO07]{DBLP:journals/siamcomp/KhotKMO07}
Subhash Khot, Guy Kindler, Elchanan Mossel, and Ryan O'Donnell.
\newblock Optimal inapproximability results for {MAX-CUT} and other 2-variable
  {CSPs}?
\newblock {\em {SIAM} J. Comput.}, 37(1):319--357, 2007.

\bibitem[KN97]{KushilevitzNisan97}
Eyal Kushilevitz and Noam Nisan.
\newblock {\em Communication Complexity}.
\newblock Cambridge University Press, 1997.

\bibitem[KNR99]{DBLP:journals/cc/KremerNR99}
Ilan Kremer, Noam Nisan, and Dana Ron.
\newblock On randomized one-round communication complexity.
\newblock {\em Comput. Complex.}, 8(1):21--49, 1999.

\bibitem[KNW10]{DBLP:conf/pods/KaneNW10}
Daniel~M. Kane, Jelani Nelson, and David~P. Woodruff.
\newblock An optimal algorithm for the distinct elements problem.
\newblock In {\em {PODS}}, pages 41--52. {ACM}, 2010.

\bibitem[Lam11]{lammersen2011approximation}
Christiane Lammersen.
\newblock {\em Approximation Techniques for Facility Location and Their
  Applications in Metric Embeddings}.
\newblock PhD thesis, Dissertation, Dortmund, Technische Universit{\"a}t, 2010,
  2011.

\bibitem[LS08]{DBLP:conf/esa/LammersenS08}
Christiane Lammersen and Christian Sohler.
\newblock Facility location in dynamic geometric data streams.
\newblock In {\em {ESA}}, volume 5193 of {\em Lecture Notes in Computer
  Science}, pages 660--671. Springer, 2008.

\bibitem[LSS09]{DBLP:conf/wads/LammersenSS09}
Christiane Lammersen, Anastasios Sidiropoulos, and Christian Sohler.
\newblock Streaming embeddings with slack.
\newblock In {\em {WADS}}, volume 5664 of {\em Lecture Notes in Computer
  Science}, pages 483--494. Springer, 2009.

\bibitem[MMR19]{MMR19}
Konstantin Makarychev, Yury Makarychev, and Ilya~P. Razenshteyn.
\newblock Performance of {J}ohnson-{L}indenstrauss transform for $k$-means and
  $k$-medians clustering.
\newblock In {\em {STOC}}, pages 1027--1038. {ACM}, 2019.

\bibitem[MRWZ20]{MahabadiRWZ20}
Sepideh Mahabadi, Ilya~P. Razenshteyn, David~P. Woodruff, and Samson Zhou.
\newblock Non-adaptive adaptive sampling on turnstile streams.
\newblock In {\em Proccedings of the 52nd Annual {ACM} Symposium on Theory of
  Computing, {STOC} 2020}, pages 1251--1264. {ACM}, 2020.
\newblock \href {https://doi.org/10.1145/3357713.3384331}
  {\path{doi:10.1145/3357713.3384331}}.

\bibitem[MW10]{MW10}
Morteza Monemizadeh and David~P. Woodruff.
\newblock 1-pass relative-error {$L_p$}-sampling with applications.
\newblock In {\em Twenty-First Annual ACM-SIAM Symposium on Discrete
  Algorithms}, SODA '10, page 1143–1160. SIAM, 2010.
\newblock \href {https://doi.org/10.1137/1.9781611973075.92}
  {\path{doi:10.1137/1.9781611973075.92}}.

\bibitem[RV07]{DBLP:journals/jacm/RudelsonV07}
Mark Rudelson and Roman Vershynin.
\newblock Sampling from large matrices: An approach through geometric
  functional analysis.
\newblock {\em J. {ACM}}, 54(4):21, 2007.

\bibitem[Sch00]{Schulman00}
L.~J. Schulman.
\newblock Clustering for edge-cost minimization.
\newblock In {\em 32nd Annual ACM Symposium on Theory of Computing}, pages
  547--555. ACM, 2000.
\newblock \href {https://doi.org/10.1145/335305.335373}
  {\path{doi:10.1145/335305.335373}}.

\bibitem[WY22]{WY22}
David~P. Woodruff and Taisuke Yasuda.
\newblock High-dimensional geometric streaming in polynomial space.
\newblock In {\em {FOCS}}. {IEEE}, 2022.
\newblock To Appear.

\bibitem[Zar11]{DBLP:journals/algorithmica/Zarrabi-Zadeh11}
Hamid Zarrabi{-}Zadeh.
\newblock An almost space-optimal streaming algorithm for coresets in fixed
  dimensions.
\newblock {\em Algorithmica}, 60(1):46--59, 2011.
\newblock \href {https://doi.org/10.1007/s00453-010-9392-2}
  {\path{doi:10.1007/s00453-010-9392-2}}.

\end{thebibliography}
    \bibliographystyle{alphaurl}
    \begin{appendices}
\section{Lower Bounds Based on INDEX}
\label{sec:lb_light}

\begin{definition}[INDEX Problem]
    Alice is given a message $x \in \{0, 1\}^n$,
    and Bob is given an index $i \in [n]$.
    Alice can send Bob exactly one message $M$, and Bob needs to use his input $i$ and this message $M$, to compute $x_i \in \{0, 1\}$.
\end{definition}
It is well known that the INDEX problem requires $\Omega(n)$ combination to succeed with constant probability, i.e., $M = \Omega(n)$ (see e.g.,~\cite{KushilevitzNisan97,DBLP:journals/cc/KremerNR99,DBLP:journals/toc/JayramKS08}).

\begin{claim}
    \label{claim:lb_sampling_light}
    For every integer $\Delta \geq 1$, given access to a quadtree $T$ on $[\Delta]$,
    any streaming algorithm that tests with constant success probability 
    whether $X(h_{i^*}) \setminus X(h_{i^* + 1}) = \emptyset$ for every $X \subseteq [\Delta]$ and $i^*$ presented as an insertion-only point stream
    must use space $\Omega(\sqrt{\Delta})$,
    where $h_i$ is defined as in \Cref{sec:streaming} which is the level-$i$ node of $T$ that contains the most from $X$,
    and it is promised that both $|X(h_{i^*})|$ and $|X(h_{i^*+1})|$ is larger than $0.5 |X|$.
\end{claim}
\begin{proof}
    Pick the level $i^*$ in $T$ such that the squares (i.e., intervals) are of side-length $m := 10\sqrt{\Delta}$.
    We reduce to INDEX with $n := \Delta / m = \frac{1}{10}\sqrt{\Delta}$,
    corresponding to the level-$i^*$ nodes. (Note that $i^*$ is public knowledge between Alice and Bob).
    Suppose Alice receives $x \in \{0, 1\}^n$.
    For $j \in [n]$, if $x_j = 1$, insert a point at coordinate $(j - 1) \cdot m + 1$,
    which is the first coordinate in the $j$-th interval at level $i^*$,
    and do nothing if $x_j = 0$.
    Feed this input to the algorithm, and send the internal state of the algorithm to Bob.

    When Bob receives $i \in [n]$,
    Bob inserts one point to each coordinate in
    the right/larger sub-interval of the $i$-th level-$i^*$ interval,
    namely, $((i - 1)\cdot m + \frac{m}{2}, i\cdot m - 1]$.
    Clearly, $h_{i^*}$ is this $i$-th interval at level $i^*$,
    $h_{i^* + 1}$ is the right sub-interval of $h_{i^*}$,
    and both $|X(h_{i^*})|$ and $|X(h_{i^* + 1})|$ contain more than $0.5 |X|$
    points.
    Observe that $X(h_{i^*}) \setminus X(h_{i^* + 1}) = \emptyset$ if and only if $x_i = 0$. This finishes the proof.
\end{proof}

\begin{claim}
    \label{claim:lb_exact}
    For every integer $\Delta \geq 1$,
    any algorithm that with constant success probability
    computes $\MaxCut(X)$ exactly for every $X \subseteq [\Delta]$ presented as an insertion-only point stream
    must use space $\Omega(\poly(\Delta))$.
\end{claim}
\begin{proof}
In our proof, we assume an algorithm $\mathcal{A}$ returns $\eta \geq 0$,
    such that for every $X \subseteq [\Delta]$,
    \[
        \Pr[\eta = \MaxCut(X)] \geq 1 - 1 / \Delta^c,
    \]
    for some sufficiently large $c \geq 1$.
    We show such $\mathcal{A}$ must use space $\Omega(\poly(\Delta))$.
    Note that the assumption about the $1 / \Delta^c$ probability is without loss of generality.

    We reduce to INDEX with $n := \Delta^{0.1}$.
    Let $m := \frac{\Delta}{n} = \Delta^{0.9} = n^9$.
    Suppose Alice receives $x \in \{0, 1\}^n$.
    For $j \in [n]$, if $x_j = 1$, insert a point at coordinate $ (j - 1) m $.
    \footnote{Here we may insert a point at $0$ which is not in $[\Delta]$.
    This may be resolved by e.g., enlarging $\Delta$ by $1$ and shifting the coordinates, but we omit this detail for the sake of presentation.
    }
    If $x_j = 0$ do nothing.
    Alice feeds this input to $\mathcal{A}$, and sends the internal state of $\mathcal{A}$ to Bob.

    Now, suppose Bob receives $i \in [n]$.
    Bob resumes algorithm $\mathcal{A}$, and use $\mathcal{A}$ to do the following:
    for every $j \in [n]$, insert to the stream $P_j := \{ (j - 1)m + \frac{m}{2} + t \mid 1 \leq t \leq n^{3} \}$ and query the \MaxCut value.
    Restore to the initial received states of $\mathcal{A}$ after every iteration of $j$.
    This may be done, by saving the received state, insert the points $P_j$ for one $j$,
    query, and fall back to the saved state.
    Eventually, since $\mathcal{A}$ succeeds with probability $1 - 1 / \Delta^c$,
    by a union bound, with constant probability,
    all the queries are answered correctly simultaneously.
    We condition on this constant-probability event.

    Let $X := \{ i \in [n] \mid x_i \neq 0 \}$ be the support of Alice's input $x$,
    and let $X' := \{ (i - 1)m \mid i\in X \}$ be the set of points inserted by Alice. 
    Let $\widehat{X}_j := X' \cup P_j$.
    Notice that now we have the exact value of $\MaxCut(\widehat{X}_j)$ for every $j \in [n]$.
    It remains to show this information suffices to deduce the $x_i$ value, which the INDEX problem asks for.
    To this end, we need to show the following lemma, that gives the (rough)
    value of $\MaxCut(\widehat{X}_j)$,
    which is basically $\sum_{i \in X} |i - j - 0.5|$ (up to scaling and a neglectable additive error).
    \begin{lemma}
        \label{lemma:lb_cut_value}
        For every $j \in [n]$, $\MaxCut(\widehat{X}_j) \in n^{12} \sum_{i \in X} |i - j - 0.5| \pm O(n^4)$.
    \end{lemma}
    \begin{proof}
        Fix some $j$. 
        Note that $P_j \cap X' = \emptyset$.
        Observe that the cut value of the cut $(X', P_j)$ is
        \begin{align*}
            \sum_{x \in X', y \in P_j} \dist(x, y)
            &= \sum_{i \in X}\sum_{t \in [n^3]} \left|(i - 1)m  - (j - 1)m - \frac{m}{2} - t\right|  \\
            &= \sum_{i \in X}\sum_{t \in [n^3]} |(i - j - 0.5)n^9 - t| \\
            &\in \sum_{i \in X}\sum_{t \in [n^3]} |i - j - 0.5| \cdot n^9 \pm t  \\
            &\in \sum_{i \in X} |i - j - 0.5|\cdot n^{12} \pm O(n^{4}).
\end{align*}
        It suffices to show the cut $(P_j, X')$ achieves $\MaxCut(\widehat{X}_j)$.
        Define the notation $\cut(W, Z) := \sum_{x \in W, y \in Z}{\dist(x, y)}$.
        Consider a cut $(S, T)$ of $\widehat{X}_j$ such that $S \neq P_j$ and $T \neq X'$,
        and we show $\cut(S, T) < \cut(X', P_j)$.
        A useful fact which follows from the definition is that $\dist(X', P_j) \geq \frac{m}{2}$.

        \paragraph{Easy Case}
        First, we consider the (easy) case such that $S \subseteq P_j $ or $T \subseteq P_j$.
        In this case, there is some $\emptyset \neq W \subset P_j$ such that $\cut(S, T) = \cut(P_j \setminus W, X' \cup W)$. Then,
        \begin{align*}
            \cut(S, T) 
            &= \cut(P_j \setminus W, X' \cup W) 
            =  \cut(P_j \setminus W, X')  + \cut(P_j, W) \\
            &<  \cut(P_j \setminus W, X') + \frac{|P_j|^2}{4} \cdot n^3 \\
            &= \cut(P_j, X') - \cut(W, X) + \frac{n^9}{4} \\
            &\leq \cut(P_j, X') - \frac{m}{2} + \frac{n^9}{4} \\
            &< \cut(P_j, X').
        \end{align*}

        \paragraph{General Case}
        Now, we handle the remaining case where neither $S$ nor $T$ is a subset of $P_j$.
        Then $|S \cap X'| > 0$ and $|T \cap X'| > 0$.
        Combining these facts and assumptions, we have
        \begin{align*}
            &\cut(P_j, X') - \cut(S, T) \\
            =&  \cut(P_j, X') - \cut(S \cap P_j, T \cap X') - \cut(S \cap X', T \cap P_j) - \cut(S \cap X', T \cap X') - \cut(S \cap P_j, T \cap P_j) \\
            >& \cut(S\cap P_j, S \cap X') + \cut(T \cap P_j, T \cap X')
            - \frac{|X'|^2}{4} \cdot \Delta  - |P_j|^2 \cdot n^3 \\
            \geq&\ \frac{m}{2}\cdot (|S \cap X'| \cdot |S \cap P_j| + |T \cap X'| \cdot |T \cap P_j|) - \frac{n^{12}}{4} - n^9 \\
            \geq&\ \frac{m}{2} \cdot (|S \cap P_j| + |T \cap P_j|) - \frac{n^{12}}{4} - n^9 \\
            =&\ \frac{n^{12}}{4} - n^9 > 0.
        \end{align*}
        This finishes the proof of \Cref{lemma:lb_cut_value}.
    \end{proof}
    Observe that, the function $f : [n] \to \RR_+$ such that $f(j) := 2\sum_{i \in X}|i - j-0.5|$ is integer-valued,
    so $2 n^{12}\sum_{i \in X}|i - j - 0.5|$ must be a multiple of $n^{12}$.
    Therefore, by \Cref{lemma:lb_cut_value},
    we round the value of $2 \cdot \MaxCut(\widehat{X}_j)$ to the nearest multiple of $n^{12}$ (so that the additive error $O(n^4)$ in \Cref{lemma:lb_cut_value} is ignored),
    and we obtain the exact value of $2 n^{12} \sum_{i \in X}|i - j -0.5|$.
    In other words, we know the value of $f(j)$ for all $j \in [n]$.
    Finally, because of the piece-wise linear structure of $f$,
    one can make use of all the $f(j)$ values to recover all of $X$.
    This finishes the proof of \Cref{claim:lb_exact}.
\end{proof}
       \section{Dimension Reduction for \MaxCut}
\label{sec:jl}

For completeness, we prove below a dimension reduction for \MaxCut,
similarly to the one given in~\cite{DBLP:conf/wads/LammersenSS09},
with proof details appearing in~\cite{lammersen2011approximation}. 

\begin{theorem}
  \label{thm:maxcut_jl}
  Let $X \subset \mathbb R^d$ be a finite set, $0 < \epsilon, \delta < 1$
  and $\pi : \mathbb R^d \to \mathbb R^{d'}$
  be a JL Transform as in \Cref{def:jl} with target dimension $d' = O\left( \varepsilon^{-2} {\log(\epsilon^{-1}\delta^{-1})} \right) $.
  Then with probability $1 - \delta$,
  \begin{equation}
    \label{eqn:maxcut_presesrved}
    \MaxCut(\pi(X)) \in (1 \pm \epsilon)\cdot \MaxCut(X).
  \end{equation}
\end{theorem}

Our bound is not directly comparable with that in~\cite{DBLP:conf/wads/LammersenSS09,lammersen2011approximation},
since we improve the dependence on $\delta$ from $O(\delta^{-2})$ to $O(\log(\delta^{-1}))$,
but we also introduced an additional $\log(\epsilon^{-1})$ factor.
Our argument is simpler, achieved by making use of a certain formulation
of the Johnson-Lindenstrauss (JL) transform~\cite{JL84},
that differs from the one used in~\cite{DBLP:conf/wads/LammersenSS09,lammersen2011approximation}. 

The JL transform formulation that we use (see \Cref{def:jl})
is similar to the one previously used in~\cite{MMR19},
to obtain dimension reduction results for clustering problems.
Its second property (the expectation bound),
may not hold for all constructions of the JL transform
(and was not used in~\cite{DBLP:conf/wads/LammersenSS09,lammersen2011approximation}),
but it can be realized by a random matrix with independent sub-Gaussian entries.

\begin{definition}[JL Transform~\cite{JL84,MMR19}]
  \label{def:jl}
  For every integer $d, d' \geq 1$,
  there exists a randomized mapping $\pi : \mathbb{R}^d \to \mathbb{R}^{d'}$
  such that for all $x \neq y \in \mathbb{R}^d$,
  \begin{align}
    \Pr_\pi & \Big[\dist(\pi(x),\pi(y)) \not\in (1 \pm \epsilon) \cdot \dist(x, y) \Big]
              \le e^{-C d' \varepsilon^2} \label{eqn:jl_pr}
    \\
    \E_\pi & \left[\max\left\{ \left|\frac{\dist(\pi(x), \pi(y))}{\dist(x, y)} - 1\right| - \varepsilon, 0\right\} \right]
             \le e^{-C d' \varepsilon^2}, \label{eqn:jl_exp}
  \end{align}
  for some universal constant $C > 0$.
\end{definition}

Our version is similar but not identical to that in~\cite{MMR19},
because we introduce an absolute value (to bound the error from both sides),
but the proof is similar.

\begin{proof}[Proof of \Cref{thm:maxcut_jl}]
  Observe that it suffices to show with probability $1 - \delta$,
  \begin{equation}
    \sum_{x,y \in X} |\dist(\pi(x),\pi(y)) - \dist(x, y)| \leq  O(\varepsilon)\sum_{x,y \in X} \dist(x, y).
  \end{equation}
  Let $P := \left\{ (x ,y) \in X : |\dist(\pi(x), \pi(y)) - \dist(x, y)| > \epsilon \dist(x, y)   \right\}$
  be the set of ``bad'' point pairs whose distances are distorted by more than $\epsilon$ error.
  Since by definition we have (with probability $1$)
  \begin{align*}
    \sum_{(x, y) \in X\times X \setminus P} |\dist(\pi(x), \pi(y)) - \dist(x, y)| \leq \epsilon \cdot \sum_{x, y \in X}\dist(x, y),
  \end{align*}
  hence, it remains to show the following holds with probability $ 1- \delta$
  \begin{equation}
    \label{eqn:bad_pairs}
    \sum_{(x, y) \in P} |\dist(\pi(x), \pi(y)) - \dist(x, y)| \leq \epsilon \sum_{x, y \in X} \dist(x, y).
  \end{equation}
  By the guarantee of \Cref{def:jl}, we have
  \begin{align*}
    \quad &\E[\sum_{(x, y) \in P} |\dist(\pi(x), \pi(y)) - \dist(x, y)| - \epsilon \dist(x, y)] \\
    = &\E[\sum_{x, y \in P}\max\{|\dist(\pi(x), \pi(y)) - \dist(x, y)| -\epsilon \dist(x, y), 0\}] \\
    \leq &\E[\sum_{x, y \in X}\max\{|\dist(\pi(x), \pi(y)) - \dist(x, y)| -\epsilon \dist(x, y), 0\}] \\
    \leq &e^{-Cd'\epsilon^2} \sum_{x, y \in X} \dist(x, y)
    \leq \epsilon \delta \sum_{x, y \in X} \dist(x, y).
  \end{align*}
Therefore, by Markov's inequality, we conclude that \eqref{eqn:bad_pairs} holds with probability $1 - \delta$.
  This finishes the proof.
\end{proof}

\Cref{thm:maxcut_jl} implies the corollary below about a streaming algorithm
that reports an encoding of a near-optimal cut (and not just its value).
The most natural way to report a cut of $X$
is to somehow represent of a $2$-partition of $X$,
but this is not possible because that contains $X$ itself,
which requires $\Omega(n)$ bits to store.
Instead, we let the algorithm report
a function $f : \mathbb{R}^d \to \{0, 1\}$ (using some encoding),
and then $f$ implicitly defines the cut $(X\cap f^{-1}(0) , X\cap f^{-1}(1))$.
In other words, the algorithm essentially reports an ``oracle''
that does not know $X$,
but can determine, for each input point $x\in X$, its side in the cut.
This formulation was suggested by~\cite{DBLP:conf/stoc/FrahlingS05},
and in fact we rely on their solution and combine it with our dimension reduction.

\begin{corollary}[Cut Oracle]
  \label{cor:cut_oracle}
  There is a randomized streaming algorithm that,
  given $0 < \epsilon < 1/2$, integers $\Delta, d \geq 1$,
  and an input dataset $X \subseteq [\Delta]^d $ presented as a dynamic stream,
  the algorithm uses space $\exp(\poly(\epsilon^{-1}))\poly(d \log \Delta)$,
  and reports (an encoding of) a mapping $f : \mathbb{R}^d \to \{0, 1\}$,
  such that with constant probability (at least $2 / 3$),
  $\cut_X(X\cap f^{-1}(0)) \geq (1 - \epsilon) \cdot \MaxCut(X)$.
\end{corollary}
\begin{proof}
  As noted in~\cite{DBLP:conf/stoc/FrahlingS05}, there exists an algorithm $\mathcal{A}$ that finds an $f$ with the same guarantee and failure probability,
  except that the space usage is $\epsilon^{-O(d)} \cdot \poly(\log \Delta)$.
  Hence, we can use this $\mathcal{A}$ as a black with \Cref{thm:maxcut_jl}
  to conclude the theorem.

  Specifically, let $\pi : \mathbb{R}^d \to \mathbb{R}^{d'}$
  such that $d' = O(\epsilon^{-2}\log(\epsilon^{-1}))$
  be a mapping that satisfies \Cref{thm:maxcut_jl}.
  Then, for every update of point $x \in [\Delta]^d$ in the stream,
  we map it to $\pi(x)$ and feed it to $\mathcal{A}$.
  When the stream terminates, we use $\mathcal{A}$ to compute an $f' : \mathbb{R}^{d'} \to \{0, 1\}$.
  Then, to define the final $f : \mathbb{R}^d \to \{0, 1\}$
  is defined as $f(x) := f'(\pi(x))$.
  This finishes the proof.
\end{proof}

     \end{appendices}
\end{document}